\documentclass{llncs}
\usepackage[english]{babel}
\usepackage{etex}

\usepackage[english]{babel}
\usepackage{amssymb}
\usepackage{amsmath}
\usepackage{mathtools}
\usepackage[all]{xy}

\title{A Polynomial Time Algorithm for Deciding Branching Bisimilarity on Totally Normed $\mathrm{BPA}$
}

\author{Chaodong He}
\institute{BASICS, Department of Computer Science, Shanghai Jiao Tong University}
\date{January 7, 2014}

\begin{document}

\maketitle

\begin{abstract}
Strong bisimilarity on normed BPA is polynomial-time decidable, while weak bisimilarity on totally normed BPA is $\mathrm{NP}$-hard.  It is natural to ask where the computational complexity of branching bisimilarity on totally normed BPA lies.  This paper confirms that this problem is polynomial-time decidable.  To our knowledge, in the presence of silent
transitions, this is the first bisimilarity checking algorithm on infinite state systems which runs in polynomial time.  This result spots an instance in which branching
bisimilarity and weak bisimilarity are both decidable but lie in different complexity classes (unless $\mathrm{NP}=\mathrm{P}$), which is not known before.

The algorithm takes the partition refinement approach and the final implementation can be thought of as a generalization of the previous algorithm of Czerwi\'{n}ski and Lasota~\cite{DBLP:conf/fsttcs/CzerwinskiL10,CzerwinskiPhD}.  However, unexpectedly, the correctness of the algorithm cannot be directly generalized from previous works, and the correctness  proof turns out to be subtle.  The proof depends on the existence of a carefully defined refinement operation fitted for our algorithm and the proposal of elaborately developed techniques, which are quite different from previous works.
\end{abstract}

\section{Introduction}

Basic process algebra (BPA)~\cite{Baeten:1991:PA:103272} is a fundamental model of infinite state systems, with its famous counterpart in the theory of formal languages: context free grammars in Greibach normal forms, which generate the entire context free languages.
In 1987, Baeten, Bergstra and Klop~\cite{BaetenBergstraKlop1987,DBLP:journals/jacm/BaetenBK93} proved a
surprising result that strong bisimilarity on normed BPA is decidable.  This result is in sharp contrast to the classical fact that language equivalence is undecidable for context free grammar~\cite{Hopcroft:1990:IAT:574901}.
After this remarkable discovery, decidability and complexity issues of bisimilarity checking on infinite state systems have been intensively investigated. See~\cite{DBLP:conf/concur/JancarM99,Burkart00verificationon,DBLP:journals/iandc/MollerSS04,Srba2004,Kucera2006}
for a number of surveys.

As regards to the strong bisimilarity checking on normed BPA,
H\"uttel and Stirling~\cite{DBLP:conf/lics/HuttelS91} improved
the result of Baeten, Bergstra and Klop using a more simplified proof by relating
the strong bisimilarity of two normed BPA processes to the existence of a successful
tableau system.  Later, Huynh and Tian~\cite{DBLP:journals/tcs/HuynhT94} showed
that the problem is in $\Sigma_{2}^{\mathrm{P}}$, the second level of the polynomial hierarchy.
Before long, another significant discovery was made by Hirshfeld,
Jerrum and Moller~\cite{DBLP:journals/tcs/HirshfeldJM96} who showed that the problem can even be decided in  polynomial
time, with the complexity $\mathcal{O}(N^{13})$.
The running time was later improved~\cite{DBLP:conf/mfcs/LasotaR06,DBLP:conf/fsttcs/CzerwinskiL10}. All these algorithms take the approach of partition refinement, relying on the unique decomposition property and some
efficient way of equality checking on compressed long strings.
It deserves special mention that Czerwi\'{n}ski and Lasota~\cite{DBLP:conf/fsttcs/CzerwinskiL10} create a different refinement scheme. This refinement scheme was previously used in developing an polynomial-time algorithm of checking strong bisimilarity on normed basic parallel processes (normed $\mathrm{BPP}$)~\cite{DBLP:journals/mscs/HirshfeldJM96}. In this way  Czerwi\'{n}ski and Lasota improve the running time to $\mathcal{O}(N^5)$.  Hitherto, the best algorithm was reported in~\cite{CzerwinskiPhD}, whose running time is  $\mathcal{O}(N^4\mathrm{polylog}(N))$.

In the presence of silent actions the picture is less clear.  Even the decidability for weak bisimilarity is still open.  A remarkable discovery is made by Fu~\cite{DBLP:conf/icalp/Fu13} recently that branching bisimilarity~\cite{GlabbeekW96}, a standard refined alternative of weak bisimilarity, is decidable on normed BPA. Very recently, Czerwi\'{n}ski and Jan\v{c}ar confirm this problem to be in $\mathrm{NEXPTIME}$~\cite{DBLP:journals/corr/CzerwinskiJ14}.  The current best lowerbound for weak bisimilarity  is the $\mathrm{EXPTIME}$-hardness established by Mayr~\cite{Mayr2005}, whose proof can be slightly modified to show
the $\mathrm{EXPTIME}$-hardness for branching bisimilarity as well.

In retrospect one cannot help thinking that more attention should
have been paid to the branching bisimilarity. Going back to the original motivation
to equivalence checking, one would agree that a specification $spec$ normally
contains no silent actions because silent actions are about how-to-do. It follows
that $spec$ is weakly bisimilar to an implementation $impl$ if and only if $spec$ is branching
bisimilar to $impl$ (Theorem~5.8.18 in~\cite{Baeten:1991:PA:103272}).  In addition, in majority of practical examples, the branching
bisimilarity and the weak bisimilarity coincide.  What these observations tell us is that as far as verification is
concerned the branching bisimilarity ought to play a bigger role than the weak
bisimilarity, especially in the situations where branching bismilarity is easily decided.

One major difficulty of checking weak or branching bisimilarity on normed $\mathrm{BPA}$ stems from the lack of nice structural properties such as unique decomposition property.  By forcing the final action of every process to be observable, we have an important subset of normed $\mathrm{BPA}$, called totally normed $\mathrm{BPA}$, in which unique decomposition property still holds for branching bisimilarity.  The bisimilarity checking on totally normed $\mathrm{BPA}$ also has a long history.
In 1991, H\"uttel~\cite{DBLP:conf/cav/Huttel91} repeated the tableau construction developed in~\cite{DBLP:conf/lics/HuttelS91} for branching
bisimilarity on totally normed BPA. Although H\"uttel's construction is not sound for weak bisimilarity, the relevant decidability can also be established~\cite{DBLP:journals/entcs/Hirshfeld96}.  For the lower bound, $\mathrm{NP}$-hardness is established by
St\v{r}\'{i}brn\'{a}~\cite{DBLP:journals/entcs/Stribrna98} for weak bisimilarity via a reduction from the knapsack problem.
By inspecting St\v{r}\'{i}brn\'{a}'s proof, we are aware that the $\mathrm{NP}$-hardness still holds for any other bisimilarity, such as delay bisimilarity,  $\eta$-bisimilarity, and even  quasi-branching bisimilarity~\cite{GlabbeekW96}, except for  branching bisimilarity.
The requirement of branching bisimilarity that change-of-state silent actions must be explicitly bisimulated makes it impossible to realize nondeterminism by designing some gadgets via a bisimulation game.   These crucial observations inspire us to rethink the possibility of designing more efficient algorithm for the problem of checking branching bisimilarity on totally normed $\mathrm{BPA}$.

The paper provides a polynomial time algorithm for checking branching bisimilarity on totally normed $\mathrm{BPA}$.  Therefore an instance is spotted that branching bisimilarity and weak bisimilarity are both decidable but lie in different complexity classes.

For brevity, in the rest of this paper, `branching bisimilarity' will usually be referred to as `bisimilarity'.   We avoid using the term `strong bisimilarity', since the strong bisimilarity can be interpreted as the bisimilarity for `realtime' processes. A realtime process is a process which can perform no silent action.

The algorithm developed in this paper takes a similar partition refinement approach and the framework adopted in~\cite{DBLP:conf/fsttcs/CzerwinskiL10,CzerwinskiPhD}, which was designed to decide bisimilarity for realtime normed $\mathrm{BPA}$. This algorithm is called CL algorithm in this paper.  The final efficient implementation of our algorithm is a generalized version of CL algorithm in the sense that, for realtime systems, our algorithm and CL algorithm are essentially the same.

Our algorithm heavily relies on the technique of dynamic programming, which makes our implementation has the same computational complexity as CL algorithm.   Although our algorithm seems very similar to the previous one,  the technical details, including the definition of expansion and refinement operation, the theoretical development of its correctness are quite difficult than the previous CL algorithm.

Without doubt, the consecutive silent transitions in the definition of branching bisimilarity cause severe problems in two aspects: the correctness and the efficiency. Note that the totally normedness guarantees that the number of consecutive silent actions are bounded by the number of constants. It is not hard to use this observation, together with the game theoretical view of branching bisimilarity, to design an algorithm which runs in polynomial space.   However, the consecutive silent actions, which cause nondeterminism, did make checking branching bisimulation property take exponential time if the naive way was taken.  The only way to overcome this difficulty is a proper usage of the technique of dynamic programming.  When consecutive silent actions are eliminated by means of dynamic programming, we have a severe problem: why is the resulting algorithm still correct?

In the situation of CL algorithm for realtime normed $\mathrm{BPA}$, there is a pre-defined refinement operation, $\mathsf{Ref}(\equiv)$. In that situation, we had a canonical definition of expansion and a canonical definition of relative decreasing bisimilarity (in our terminology). The final refinement operation $\mathsf{Ref}(\equiv)$ was defined as the decreasing bisimilarity wrt.~the expansion of $\equiv$.  The refined equivalence relation was then constructed by a greedy algorithm, in each step of which two memberships were efficiently tested.  Therefore, the correctness of CL algorithm was comparatively obvious.

Unfortunately it is unlikely, if not impossible, to generate the above proof structure for CL algorithm to our algorithm, because there is no clear way to define the expansion relation like that in CL algorithm. Note that the expansion relation should be both correct and efficient. We had several aborted attempts before finally we decided to take some other ways.

The correctness of CL algorithm depends on a clearly defined refinement operation which relies on two steps of operations: the expansion operation and the relative decreasing bisimilarity.  Our crucial insight is that, there is no need to separate these two steps of operations.    The main technical line is briefly outlined below.  It takes several stages:
\begin{itemize}
\item
 At first, for realtime systems, we define  the refinement operations by a way of combining the two steps of operation which was taken in CL algorithm into a cohesive whole.  In this way, we have noticed that we defines exactly the same refinement operation as that in CL algorithm.

\item
Then,
the refinement operation defined in the above way  is smoothly generated for the style of branching bisimilarity. In this stage, our attention is centred on the property of the refined relation. The efficiency is never cared about.  We prove that the refinement operation preserves  congruence and the unique decomposition property.

\item
Then a characterization theorem is established for the refined congruence.  In this characterization, the consecutive silent actions are completely eliminated. Thus the problem of efficiency is mainly solved. Using this characterization, the correctness proof for realtime systems can be obtained.  But for systems with silent actions, it is not enough.

\item
Finally, the proof is finished by developing a simpler characterization  which corresponds to our algorithm directly. In this stage,  a special property of branching bisimilarity for processes in prime decomposition turns out to be quite useful.
\end{itemize}

The rest of the paper is organized as follows. Section~\ref{sec:Preliminaries} lays down the preliminaries. Section~\ref{sec:finite_representations} focuses on the unique decomposition property for branching bisimilarity on totally normed BPA. Then we describe our algorithm in Section~\ref{sec:naive-algorithm}. The suitable definition of refinement steps are discussed in Section~\ref{sec:refinement_steps},  and the correctness proof are provided in Section~\ref{sec:correctness}.  Finally, Section~\ref{sec:remark} gives additional remarks.

%%%%%%%%%%%%%%%%%%%%%%%%%%%%%%%%%%%%%%%%%%%%%%%%%%%%%%%%%555

\section{Preliminaries}\label{sec:Preliminaries}
\subsubsection{Basic Process Algebra}
A {\em basic process algebra} ($\mathrm{BPA}$) system is a triple $(\mathbf{C}, \mathcal{A}, \Delta)$, where $\mathbf{C} = \{X_1, \ldots, X_n\}$ is a finite set of process constants, $\mathcal{A}$ is a finite set of actions, and $\Delta$ is a finite set of transition rules.
The {\em processes}, ranged over by $\alpha,\beta,\gamma,\delta$,  are generated by the following grammar:
\[
\alpha \  \Coloneqq  \  \epsilon  \  \mid \  X \  \mid \  \alpha_1 \cdot \alpha_2.
\]
The syntactic equality is denoted by $=$.
We assume that the sequential composition $\alpha_1\cdot\alpha_2$ is associative up to $=$ and $\epsilon \cdot \alpha = \alpha \cdot \epsilon = \alpha$. Sometimes $\alpha \cdot \beta$ is shortened as $\alpha\beta$.  The set of processes is exactly $\mathbf{C}^{*}$, the strings over $\mathbf{C}$.
There can be a special symbol $\tau$ in $\mathcal{A}$ for silent transition.  Typically, $\ell$ is used to denote actions, while $a$ are used to denote visible (i.e. non-silent) actions.
The transition rules in $\Delta$ are of the form $X \stackrel{\ell}{\longrightarrow} \alpha$.
The following labelled transition rules define the operational semantics of the processes.
\[
\begin{array}{c}
    \cfrac{X\stackrel{\ell}{\longrightarrow}P\in\Delta}{X\stackrel{\ell}{\longrightarrow}\alpha} \
                    \qquad \cfrac{ \alpha\stackrel{\ell}{\longrightarrow}\alpha'}{\alpha\cdot\beta \stackrel{\ell}{\longrightarrow} \alpha' \cdot \beta}
 \end{array}
\]
The operational semantics is structural, meaning that $\alpha \cdot\beta\stackrel{\ell}{\longrightarrow}\alpha'\cdot\beta$ whenever $\alpha\stackrel{\ell}{\longrightarrow}\alpha'$.
We write $\Longrightarrow $ for the reflexive transitive closure of $\stackrel{\tau}{\longrightarrow}$, and $\stackrel{\widehat{\ell}}{\Longrightarrow}$ for $\Longrightarrow\stackrel{\ell}{\longrightarrow}\Longrightarrow$ if $\ell\ne\tau$ and for $\Longrightarrow$ otherwise.

A process $\alpha$ is {\em normed} if $\alpha \stackrel{\ell_1}{\longrightarrow} \dots  \stackrel{\ell_n}{\longrightarrow} \epsilon$ for some $\ell_1, \dots, \ell_n$. A process $\alpha$ is {\em totally normed} if it is normed, and moreover, $\ell_n \neq \tau$ whenever $\alpha  \stackrel{\ell_1}{\longrightarrow}  \dots  \stackrel{\ell_n}{\longrightarrow} \epsilon$.
A $\mathrm{BPA}$ definition  $(\mathbf{C}, \mathcal{A}, \Delta)$ is (totally) normed if all processes defined in it are (totally) normed. We write (t)(n)BPA for the (totally)(normed) basic process algebra model. In other words, a $\mathrm{tnBPA}$ system is a $\mathrm{nBPA}$ system in which rules of the form $X  \stackrel{\tau}{\longrightarrow} \epsilon$ are forbidden.

We call a BPA system {\em realtime} if $\tau \not\in \mathcal{A}$. That is to say, a realtime system can not perform silent actions.
Clearly, realtime totally normed BPA is exactly realtime normed BPA.

\subsubsection{Bisimulations and Bisimilarities}
In the presence of silent actions two well known process equalities are the branching bisimilarity~\cite{GlabbeekW96} and the weak bisimilarity~\cite{Milner1989}.

\begin{definition}\label{def:beq}
Let $\mathcal{R}$ be a relation on processes. $\mathcal{R}$ is a {\em branching bisimulation},  if the following hold whenever $\alpha \mathcal{R} \beta$:
\begin{enumerate}
\item
If $\alpha \stackrel{\ell}{\longrightarrow} \alpha'$, then either
 \begin{enumerate}
  \item
  $\ell = \tau$ and $\alpha'\mathcal{R}\beta$; or

 \item
    $\beta \Longrightarrow \beta''\stackrel{\ell}{\longrightarrow} \beta'$ and $\alpha'\mathcal{R} \beta'$ and $\alpha\mathcal{R}\beta''$ for some $\beta',\beta''$.
 \end{enumerate}

\item
If  $\beta\stackrel{\ell}{\longrightarrow}\beta'$, then either
 \begin{enumerate}
 \item
    $\ell=\tau$ and
    $\alpha\mathcal{R}\beta'$; or

 \item
 $\alpha\Longrightarrow \alpha''\stackrel{\ell}{\longrightarrow}
    \alpha'$ and $\alpha'\mathcal{R} \beta'$ and $\alpha'' \mathcal{R} \beta$ for some $\alpha',\alpha''$.
    \end{enumerate}
\end{enumerate}
The {\em branching bisimilarity}  $\simeq$ is the largest branching bisimulation.
\end{definition}

\begin{definition}
A relation $\mathcal{R}$ is a {\em weak bisimulation} if the following are valid:
\begin{enumerate}
\item
Whenever $\alpha \mathcal{R}\beta$ and $\alpha \stackrel{\ell}{\longrightarrow} \alpha'$, then $\beta\stackrel{\widehat{\ell}}{\Longrightarrow}\beta'$ and $\alpha'\mathcal{R}\beta'$ for some $\beta'$.
\item
Whenever $\alpha\mathcal{R}\beta$ and $\beta \stackrel{\ell}{\longrightarrow} \beta'$, then $\alpha\stackrel{\widehat{\ell}}{\Longrightarrow}\alpha'$ and $\alpha'\mathcal{R} \beta'$ for some $\alpha'$.
\end{enumerate}
The {\em weak bisimilarity} $\approx$ is the largest weak bisimulation.
\end{definition}

Both $\simeq$ and $\approx$ are congruence relations for (t)nBPA. We remark that transitivity of $\simeq$ is not straightforward according to Definition~\ref{def:beq}, because the branching bisimulation $\mathcal{R}$ defined in Definition~\ref{def:beq} need not be transitive~\cite{DBLP:journals/ipl/Basten96}.  To solve this problem, van Glabbeek and Weijland~\cite{GlabbeekW96} introduce a slightly different notion called {\em semi-branching bisimulation}.

\begin{definition}\label{def:semi-beq}
Let $\mathcal{R}$ be a relation on processes. $\mathcal{R}$ is a {\em semi-branching bisimulation}  if the following hold whenever $\alpha \mathcal{R} \beta$:
\begin{enumerate}
\item
If $\alpha \stackrel{\ell}{\longrightarrow} \alpha'$, then either
 \begin{enumerate}
  \item
  $\ell = \tau$ and $\beta \Longrightarrow \beta'$ for some $\beta'$ such that $\alpha\mathcal{R}\beta'$ and $\alpha'\mathcal{R}\beta'$; or

 \item
    $\beta \Longrightarrow \beta''\stackrel{\ell}{\longrightarrow} \beta'$ and $\alpha'\mathcal{R} \beta'$ and $\alpha\mathcal{R}\beta''$ for some $\beta',\beta''$.
 \end{enumerate}

\item
If  $\beta\stackrel{\ell}{\longrightarrow}\beta'$, then either
 \begin{enumerate}
 \item
    $\ell=\tau$ and
     $\alpha \Longrightarrow \alpha'$ for some $\alpha'$ such that $\alpha'\mathcal{R}\beta$ and $\alpha'\mathcal{R}\beta'$; or

 \item
 $\alpha\Longrightarrow \alpha''\stackrel{\ell}{\longrightarrow}
    \alpha'$ and $\alpha'\mathcal{R} \beta'$ and $\alpha'' \mathcal{R} \beta$ for some $\alpha',\alpha''$.
 \end{enumerate}
\end{enumerate}
\end{definition}

Then it is easy to establish the following facts:
\begin{enumerate}
\item
A branching bisimulation is a semi-branching bisimulation.

\item
A semi-branching bisimulation is transitive.

\item
The largest semi-branching bisimulation is an equivalence.

\item
The largest semi-branching bisimulation is a branching bisimulation.
\end{enumerate}

Now the largest semi-branching bisimulation is the same as $\simeq$, the largest branching bisimulation.

If the involved system is realtime,  then the branching bisimilarity and the weak bisimilarity are coincident. They are called the {\em strong bisimilarity}  and are denoted by $\sim$ in literature.
In this paper, branching bisimilarity is often abbreviated as {\em bisimilarity}. If the system is realtime, we also use the term bisimilarity to indicate strong bisimilarity.  However, we tend to use the term `branching bisimilarity' in the situation of discussing on its relationship with weak bisimilarity.

The following lemma, first noticed by van Glabbeek and Weijland~\cite{GlabbeekW96}, plays a fundamental role in the study of bisimilarity.
\begin{lemma}\label{computation-lemma}
If $\alpha \Longrightarrow \alpha'\Longrightarrow  \alpha'' \simeq \alpha$ then $\alpha'\simeq \alpha$.
\end{lemma}

Let $\approxeq$ be a process equivalence.
A silent action $\alpha\stackrel{\tau}{\longrightarrow}\alpha'$ is {\em state-preserving} with regards to $\approxeq$ if $\alpha'\approxeq \alpha$;
it is {\em change-of-state} with regards to $\approxeq$ if $\alpha'\not\approxeq \alpha$.
Branching bisimilarity strictly refines weak bisimilarity in the sense that only state-preserving silent actions can be ignored; a change-of-state must be explicitly bisimulated.
Suppose that $\alpha \simeq \beta$ and $\alpha \stackrel{\ell}{\longrightarrow} \alpha'$ is matched by the transition sequence $\beta \stackrel{\tau}{\longrightarrow} \dots \stackrel{\tau}{\longrightarrow} \beta^{i} \stackrel{\tau}{\longrightarrow} \dots \stackrel{\tau}{\longrightarrow}\beta'' \stackrel{\ell}{\longrightarrow} \beta'$.
By definition one has $\alpha \simeq \beta''$.
It follows from Lemma~\ref{computation-lemma} that $\alpha \simeq \beta^{i}$, meaning that all silent actions in $\beta\Longrightarrow \beta''$ are necessarily state-preserving.  This property fails for the weak bisimilarity as the following example demonstrates.
\begin{example}\label{example}
Consider the $\mathrm{tnBPA}$ system whose rules are defined by
\[
\{X \stackrel{b}{\longrightarrow} \epsilon, \ X \stackrel{\tau}{\longrightarrow} X', \ X'\stackrel{a}{\longrightarrow} \epsilon, \ X \stackrel{a}{\longrightarrow} \epsilon;\ Y \stackrel{b}{\longrightarrow} \epsilon, \ Y\stackrel{\tau}{\longrightarrow} Y', \ Y' \stackrel{a}{\longrightarrow} \epsilon\}.
\]
One has $X \approx Y$.
However $ X \not\simeq Y$ since $Y\not\simeq Y'$.
\end{example}

\subsubsection{Norm}
Given an tnBPA system $(\mathbf{C}, \mathcal{A}, \Delta)$.
We relate a natural number $\mathtt{norm}(X)$, the {\em norm} of $X$, to every constant $X$, defined as the least $k$ such that $X \Longrightarrow \stackrel{a_1}{\longrightarrow} \Longrightarrow \dots  \Longrightarrow\stackrel{a_k}{\longrightarrow} \epsilon$.  Silent actions contribute zero to norm.   $\mathtt{norm}$ is extended to processes by taking $\mathtt{norm}(\epsilon) = 0$ and $\mathtt{norm}(X \cdot \alpha) = \mathtt{norm}(X) + \mathtt{norm}(\alpha)$.

\begin{lemma}\label{lem:normal_one}
In a $\mathrm{tnBPA}$ system, $\mathtt{norm}(\alpha) = 0$ if and only if $\alpha = \epsilon$.
\end{lemma}

A transition $\alpha \stackrel{\ell}{\longrightarrow} \alpha'$ is {\em decreasing}, denoted by $\alpha \stackrel{\ell}{\longrightarrow}_{\mathrm{dec}} \alpha'$ if either
$\ell \neq \tau$ and $\mathtt{norm}(\alpha) = \mathtt{norm}(\alpha') + 1$, or
$\ell = \tau$ and $\mathtt{norm}(\alpha) = \mathtt{norm}(\alpha')$.  The notion of
decreasing transitions formalizes the intuition that a transition can be extended to a path which witnesses the norm of $\alpha$.

\subsubsection{Standard Input}

For technical convenience, we require the input $\mathrm{tnBPA}$ system $(\mathbf{C}, \mathcal{A}, \Delta)$ to be {\em standard}, which have the following two additional properties:
\begin{enumerate}
\item
The constants in $\mathbf{C} = \{ X_i\}_{i=1}^{n} $ are ordered by non-decreasing norm, that is:
\[
\mathtt{norm}(X_1) \leq \mathtt{norm}(X_2) \leq \ldots \leq \mathtt{norm}(X_n).
\]

\item
 Let $\mathbf{C}_i$ be the set $\{X_1, X_2, \ldots, X_i\}$ for $i=0,1,\ldots, n$. In particular, $\mathbf{C}_0 = \emptyset$ and $\mathbf{C}_n = \mathbf{C}$. Assume $X_i \stackrel{\ell}{\longrightarrow}_{\mathrm{dec}} \alpha$, we need the property $\alpha \in \mathbf{C}_{i-1}^{*}$.   This property does not hold in general because of the existence of loops like $X_i \Longrightarrow X_j \Longrightarrow X_i$.  In this case we have $X_i \simeq X_j$ by Lemma~\ref{computation-lemma}, and we can transform the system by contracting $X_i$ and $X_j$ into one constant (removing $X_j$ and substituting all occurrences of $X_j$ in $\Delta$ by $X_i$) and eliminating the loop rules.   All loops can be eliminated in this way.  (By totally normedness, $X \stackrel{\tau}{\Longrightarrow}_{\mathrm{dec}} X \cdot Y$ is impossible.)  Afterwards, we specify a partial order ${\preceq} \in \mathbf{C} \times \mathbf{C}$ such that $X \prec X'$ if and only if either $\mathtt{norm}(X) < \mathtt{norm}(X')$ or $X' \Longrightarrow_{\mathrm{dec}} X$. Then the order of constants are chosen to be any total order which extends $\prec$. These works can be done by computing the `dependency graph' and then calling an algorithm for topological sort.
\end{enumerate}

The size of a $\mathrm{tnBPA}$ system $(\mathbf{C}, \mathcal{A}, \Delta)$ is denoted by $|\Delta|$. A procedure is said to be {\em efficient} if it runs in polynomial time.  The above discussion confirms that any $\mathrm{tnBPA}$ system can be efficiently transformed to a standard one with no size growing.

\begin{lemma}\label{lem:normal_two}
For every $\mathrm{tnBPA}$ system $(\{X_1, X_2, \ldots. X_n\}, \mathcal{A}, \Delta)$, there is a standard $\mathrm{tnBPA}$ system $(\{X_1', X_2',\ldots,X_m'\}, \mathcal{A}, \Delta')$ computable in at most $\mathcal{O}(|\Delta|^2)$ time, in which $m \leq n$ and $|\Delta'| \leq  |\Delta|$.
\end{lemma}
From now on, the input $\mathrm{tnBPA}$ system is supposed to be standard, and is fixed as $(\mathbf{C}, \mathcal{A}, \Delta)$ where $\mathbf{C} = \{X_1, X_2, \ldots. X_n\}$.  We will invariantly use $n$ to denote the size of $\mathbf{C}$, and $N$ to denote the size of the related $\mathrm{tnBPA}$ system.

The problem is formally defined as follows:
\begin {center}\small
 \begin{tabular}{|rp{9.5cm}|}\hline
        Problem:   \quad  &  \textsc{Branching Bisimilarity on tnBPA} \\
        Instance:  \quad  &  A standard tnBPA system $(\mathbf{C} = \{X_i\}_{i=1}^{n}, \mathcal{A}, \Delta)$, and $\alpha, \beta \in \mathbf{C}^{*}$.  \\
        Question:  \quad  &  $\alpha \simeq \beta$? \\  \hline
    \end{tabular}
\end {center}

We restate the important property for standard systems as the following lemma.
\begin{lemma}\label{lem:decreasing_transition}
Assume $X_i \stackrel{\ell}{\longrightarrow}_{\mathrm{dec}} \alpha$, we have $\alpha \in \mathbf{C}_{i-1}^{*}$.
\end{lemma}

\subsubsection{Other Conventions}

We will always use notation $\equiv$ to denote an equivalence/congruence relation on $\mathbf{C}^{*}$.  An equivalence/congruence relation $ \equiv$ is {\em norm-preserving} if  $\mathtt{norm}(\alpha) = \mathtt{norm}(\alpha')$ whenever $\alpha \equiv \alpha'$. In this paper, all the equivalence/congruence relations are supposed to be norm-preserving. This fact is not always explicitly stated.

\section{Finite Representations}\label{sec:finite_representations}

In this section, we propose  a convenient way of representing bisimilarity and the approximating congruences \footnote{The proofs in this section is a generalization of the corresponding work for realtime normed BPA, say~\cite{DBLP:journals/tcs/HirshfeldJM96}. The readers familiar with these former works can only skim this part.}.

From the algebraic view, the set of processes of $\mathrm{tnBPA}$ is exactly the free monoid generated by $\mathbf{C}$.  The question is how to represent a congruence relation on $\mathbf{C}^{*}$.
We will show that the bisimilarity $\simeq$ is a very special congruence. Not only is it finitely generated, but it enjoys a highly structured property called {\em unique decomposition property}.

\subsection{Unique Decomposition Property of $\simeq$}

Unique decomposition property plays a central role in all the algorithms for bisimilarity checking on realtime $\mathrm{nBPA}$. This important property also holds for bisimilarity on $\mathrm{tnBPA}$.

Recall that a congruence $\equiv$ is {\em norm-preserving} if $\mathtt{norm}(\alpha) = \mathtt{norm}(\beta)$ whenever $\alpha \equiv \beta$. The following lemma is a direct consequence of Definition~\ref{def:beq}.
\begin{lemma}\label{lem:norm-preserving}
$\simeq$ is a norm-preserving congruence.
\end{lemma}
Let ${\equiv} \subseteq \mathbf{C}^{*} \times \mathbf{C}^{*}$ be an arbitrary norm-preserving congruence.
Intuitively, a constant process $X_i$ is a composite if $X_i \equiv \alpha\beta$ for some $\alpha, \beta \neq \epsilon$. In this case we also have $\mathtt{norm}(\alpha), \mathtt{norm}(\beta) < \mathtt{norm}(X_i)$ from Lemma~\ref{lem:normal_one}.
For technical convenience we will define $X_i$ to be a {\em composite} modulo $\equiv$ if $X_i \equiv \alpha$ for some $\alpha \in \mathbf{C}_{i-1}^{*}$. Otherwise, $X_i$ is called a {\em prime} modulo $\equiv$.

Let $\mathbf{P} \subseteq \mathbf{C}$ be the set of primes modulo $\equiv$.  By Lemma~\ref{lem:norm-preserving} and the well-foundedness of natural numbers, every $X \in \mathbf{C}$ has a {\em prime decomposition} $\alpha \in \mathbf{P}^{*}$ such that $X_i \equiv \alpha$.   We say that $\equiv$ has unique decomposition property, or simply $\equiv$ is {\em decompositional} if every process has exactly one prime decomposition.

It is the time to establish the unique decomposition property of $\simeq$. The following Lemma~\ref{lem:right-cancellation} and Theorem~\ref{thm:unique-decomposition} is standard, as is in the case of bisimilarity for realtime $\mathrm{nBPA}$~\cite{DBLP:journals/tcs/HirshfeldJM96}.
The {\em right cancellation} property is established first.

\begin{lemma}[Right Cancellation]\label{lem:right-cancellation}
$\alpha\gamma \simeq \beta\gamma$ entails $\alpha \simeq \beta$.
\end{lemma}

\begin{proof}
$\{(\alpha, \beta): \alpha\gamma \simeq \beta\gamma \mbox{ for some }\gamma\}$
is a bisimulation.  \qed
\end{proof}

\begin{theorem}[Unique Decomposition Property of $\simeq$]\label{thm:unique-decomposition} $\simeq$ is decompositional.
Let $X_{i_1} \ldots X_{i_p}$ and  $X_{j_1} \ldots X_{j_q}$ be two irreducible decompositions such that $X_{i_1} \ldots X_{i_p} \simeq X_{j_1} \ldots X_{j_q}$. Then,  $p=q$ and $X_{i_t} \simeq X_{j_t}$ for every $1 \leq t \leq p$.
\end{theorem}

\begin{proof}
Assume on the contrary that $X_{i_1} \ldots X_{i_p}$ and  $X_{j_1} \ldots X_{j_q}$ be two different irreducible decompositions with the least norm such that
\[
X_{i_1} \ldots X_{i_p} \simeq X_{j_1} \ldots X_{j_q} .
\]
Suppose that
\begin{equation}\label{eqn:proof_UDP1}
X_{i_1} \ldots X_{i_p} \Longrightarrow_{\mathrm{dec}} \stackrel{a}{\longrightarrow}_{\mathrm{dec}} \gamma X_{i_2} \ldots X_{i_p}.
\end{equation}
These actions must be bisimulated (matched) by
\begin{equation}\label{eqn:proof_UDP2}
X_{j_1} \ldots X_{j_q} \Longrightarrow_{\mathrm{dec}} \stackrel{a}{\longrightarrow}_{\mathrm{dec}}  \delta X_{j_2} \ldots X_{j_q}
\end{equation}
for some $\delta$ such that $\gamma X_{i_2} \ldots X_{i_p} \simeq \delta X_{j_2} \ldots X_{j_q}$.
Since the norm of $\gamma X_{i_2} \ldots X_{i_p}$ and $\delta X_{j_2} \ldots X_{j_q}$ is strictly decremented, we have $X_{i_p} \simeq X_{j_q}$ from the induction hypothesis.  Now by right cancellation lemma, $X_{i_1} \ldots X_{i_{p-1}} \simeq X_{j_1} \ldots X_{j_{q-1}}$. This contradicts with the minimum norm assumption.  \qed
\end{proof}
On the other direction,  right or left cancellation property is an implication of unique decomposition property.
\begin{lemma}\label{lemma:udp_to_cancellation}
Let $\equiv$ be decompositional. Then $\alpha\gamma \equiv \beta\gamma$ (or $\gamma\alpha \equiv \gamma\beta$) implies $\alpha \equiv \beta$.
\end{lemma}

\begin{remark}
The proof of Lemma~\ref{lem:right-cancellation} and Theorem~\ref{thm:unique-decomposition} is standard~\cite{BaetenBergstraKlop1987,DBLP:journals/tcs/HirshfeldJM96}.
Although the proof is fairly straightforward, it heavily depends  on {\em branching} bisimilarity and {\em totally} normedness. For example in the above proof when actions coming from $X_{i_1}$ in (\ref{eqn:proof_UDP1}) are matched by the actions in (\ref{eqn:proof_UDP2}),  the crucial point is that $X_{j_2}$ is never used.  This cannot be proved in the case of weak bisimilarity, or in the case without totally normedness.  We will have the following two counterexamples if branching bisimilarity is replaced by {\em weak} bisimilarity, or if the condition of totally normedness is abandoned.
\end{remark}

\begin{example}\label{example:weak_bisimilarity_decom}
This counterexample is borrowed from~\cite{DBLP:conf/cav/Huttel91}. Consider the tnBPA system $(\{X, Y, B, A\}, \{a\}, \Delta)$, with
\[
\Delta = \{ X \stackrel{a}{\longrightarrow} Y,  Y \stackrel{a}{\longrightarrow} \epsilon, Y \stackrel{\tau}{\longrightarrow} X,  A \stackrel{a}{\longrightarrow} \epsilon, A \stackrel{\tau}{\longrightarrow} B, B \stackrel{a}{\longrightarrow} \epsilon \}.
\]
Clearly, $A Y \approx B Y $ but $A \not\approx B$. Right cancellation property does not hold, neither does the unique decomposition property hold.
\end{example}

\begin{example}
Consider the nBPA system $(\{X\}, \{a\}, \Delta)$, with
\[
\Delta = \{ X \stackrel{a}{\longrightarrow} X, X \stackrel{\tau}{\longrightarrow} \epsilon\}.
\]
Clearly, $X \simeq XX \simeq XXX \simeq \ldots$.  Unique decomposition property fails in this example merely because the existence of {\em idempotent} processes.
\end{example}

\subsection{Decomposition Bases}\label{subset:decomposition_base}

A decompositional congruence over $\mathbf{C}^{*}$ can be represented by a decomposition base.  A {\em decomposition base} $\mathcal{B}$ is a pair $(\mathbf{P}, \mathbf{E})$, in which $\mathbf{P} \subseteq \mathbf{C}$ specifies the set of primes, and $\mathbf{E}$ is a finite set of equations of the form $X = \alpha_X$ for every $X \in \mathbf{C} - \mathbf{P}$ and $\alpha_X \in \mathbf{P}^{*}$.  The equation $X = \alpha_X$ realizes the fact that every composite $X$ is equal to a string of primes $\alpha_X$ which is the {\em prime decomposition} of $X$.  The congruence relation generated by $\mathcal{B}$ is denoted by $\stackrel{\mathcal{B}}{\equiv}$.

The {\em prime decomposition} of a process $\alpha$ with regard to $\mathcal{B}$ is denoted by $\mathtt{dcmp}_{\mathcal{B}}(\alpha)$.
Formally, we set $\mathtt{dcmp}_{\mathcal{B}}(X) = X$ when $X \in \mathbf{P}$, and $\mathtt{dcmp}_{\mathcal{B}}(X) = \alpha_X$ wherever the equation $X = \alpha_X$ is in $\mathbf{E}$.
The domain of $\mathtt{dcmp}_{\mathcal{B}}$ is extended to $\mathbf{C}^{*}$ naturally by setting $\mathtt{dcmp}_{\mathcal{B}}(\epsilon) = \epsilon$ and  $\mathtt{dcmp}_{\mathcal{B}}(\alpha\cdot \beta) = \mathtt{dcmp}_{\mathcal{B}}(\alpha) \cdot \mathtt{dcmp}_{\mathcal{B}}(\beta)$.

The following lemma makes checking $\alpha \stackrel{\mathcal{B}}{\equiv} \beta$ fairly easy by only computing the prime decompositions of $\alpha$ and $\beta$.

\begin{lemma}\label{lem:dcmp}
$\alpha \stackrel{\mathcal{B}}{\equiv} \beta$ if and only if $\mathtt{dcmp}_{\mathcal{B}}(\alpha) = \mathtt{dcmp}_{\mathcal{B}}(\beta)$.
\end{lemma}

In the rest of the paper, every congruence $\mathcal{B} = (\mathbf{P}, \mathbf{E})$ generated by a decomposition base $\mathcal{B}$ is assumed to be norm-preserving.  Thus we must have $\mathtt{norm}(X) = \mathtt{norm}(\alpha_X)$ if the equation $X = \alpha_X$ is in $\mathbf{E}$.

The following lemma formalizes the important observation that prime constants  do not have state-preserving silent actions.
\begin{lemma}\label{lem:tau_prime}
Let $\mathcal{B}=(\mathbf{P}, \mathbf{E})$ be a decomposition base, and $X_i \in \mathbf{P}$. Assume $X_i \stackrel{\ell}{\longrightarrow}_{\mathrm{dec}} \alpha$, we have $X_i \not  \stackrel{\mathcal{B}}{\equiv} \alpha$.
\end{lemma}
\begin{proof}
According to Lemma~\ref{lem:decreasing_transition}, $\alpha \in \mathbf{C}_{i-1}$.

\begin{enumerate}
\item
If $\ell =\tau$ and $\mathtt{norm}(X_i) = \mathtt{norm}(\alpha)$.  In this case, if we have $X_i  \stackrel{\mathcal{B}}{\equiv} \alpha$, then according to the fact that $X_i$ being prime and Lemma~\ref{lem:dcmp},  $X_i = \mathtt{dcmp}_{\mathcal{B}}(X_i) = \mathtt{dcmp}_{\mathcal{B}}(\alpha)$. This is a contradiction.

\item
If $\ell \neq \tau$ and $\mathtt{norm}(X_i) = \mathtt{norm}(\alpha) + 1$, we cannot have $X_i  \stackrel{\mathcal{B}}{\equiv} \alpha$ because $\stackrel{\mathcal{B}}{\equiv}$ is norm preserving.  \qed
\end{enumerate}

\end{proof}
The above property can be lifted from constants to processes, regarding Lemma~\ref{lemma:udp_to_cancellation}.
\begin{lemma}\label{lem:tau_prime_string}
Let $\mathcal{B}=(\mathbf{P}, \mathbf{E})$ be a decomposition base, and $\alpha \in \mathbf{P}^{*}$. Assume $\alpha \stackrel{\ell}{\longrightarrow}_{\mathrm{dec}} \gamma$, we have $\alpha \not  \stackrel{\mathcal{B}}{\equiv} \gamma$.
\end{lemma}

\begin{remark}
Algebraically, a decomposition base $\mathcal{B}$ can be understood as a {\em finite presentation} of a monoid. In fact, $\mathcal{B}$ specifies the quotient monoid $\mathbf{C}^{*} / \stackrel{\mathcal{B}}{\equiv}$. Moreover,  the unique decomposition property says that the quotient monoid $\mathbf{C}^{*} / \stackrel{\mathcal{B}}{\equiv}$ is a free monoid.
From computational point of view, $\mathcal{B}$ is a {\em string rewriting system}. Rewriting rules are exact the equations in $\mathbf{E}$ from left to right.  Strings in {\em normal forms} are exact $\mathbf{P}^{*}$, the free monoid generated by $\mathbf{P}$.  All composites can be reduced to its prime decompositions. Any $\alpha \in \mathbf{C}^{*}$ has a normal form.  Church-Rosser property is guaranteed by the unique decomposition property, which makes checking $\alpha \stackrel{\mathcal{B}}{\equiv} \beta$ fairly easy by merely rewriting $\alpha$ and $\beta$ to their normal forms.
\end{remark}

\section{Description of the Algorithm}\label{sec:naive-algorithm}
This section serves as the description of our algorithm.
The algorithm takes the partition refinement approach.  It is a generalized version of the one in~\cite{DBLP:conf/fsttcs/CzerwinskiL10}, which we call CL algorithm. However, unlike the original CL algorithm, the correctness of our algorithm is not obvious and is much more difficult to prove.  This is the reason why we describe the algorithm  before we prove its correctness.  During the description, we also show some properties and requirements which make the algorithm work. A few properties are not proved until Section~\ref{sec:refinement_steps}.

\subsection{Partition Refinements with Decomposition Bases}
In order to decide whether $\alpha \simeq \beta$, we start with
an initial congruence relation $\equiv_0$,
and iteratively refine it. The refinement operation will be denoted by $\mathsf{Ref}$. By taking $\equiv_{i+1} = \mathsf{Ref}(\equiv_i)$, we have a sequence of congruence relations
\[
{\equiv_0},  {\equiv_1} , {\equiv_2} ,  \ldots
\]
which satisfy
\[
{\equiv_0} \supseteq  {\equiv_1} \supseteq {\equiv_2} \supseteq  \ldots.
\]
The correctness of the refinement operation adopted in this paper depends on the following requirements:
\begin{enumerate}
\item
${\simeq} \subseteq {\equiv_0}$.

\item
$\mathsf{Ref}({\simeq}) = {\simeq}$.

\item
If ${\simeq} \subsetneq {\equiv}$, then $ {\simeq} \subseteq \mathsf{Ref}({\equiv}) \subsetneq {\equiv}$.
\end{enumerate}
Once the sequence becomes stable, say ${\equiv_i} =  {\equiv_{i+1}}$, we have ${\simeq} = {\equiv_i}$.

\begin{remark}
The refinement operation taken in this paper leads to a monotonic sequence $\{\equiv_i\}_{i \in \omega}$.  Namely,
\[
{\equiv_0} \supseteq  {\equiv_1} \supseteq {\equiv_2} \supseteq  \ldots.
\]
This property is not necessary in a general framework of refinement. One alternative is to replace the third requirement above by the following two:
\begin{enumerate}
\item[3'.]
$\mathsf{Ref}$ is monotone. $\mathsf{Ref}({\equiv}) \subseteq \mathsf{Ref}({\equiv}')$ whenever ${\equiv} \subseteq {\equiv}'$.

\item[4'.]
If ${\simeq} \subsetneq {\equiv}$, then $\mathsf{Ref}({\equiv}) \neq {\equiv}$.
\end{enumerate}
\end{remark}

In the algorithm,  the congruences $\simeq$ and $\equiv_i$'s are all represented by decomposition bases. That is, all the intermediate $\equiv_i$ must be decompositional congruences. In the following, we will develop an implementation of the refinement steps in polynomial time.

On the whole, the algorithm is an iteration:
\begin{center}
\small
 \begin{tabular}{|p{12cm}|}\hline\vspace{-2ex}
\begin{itemize}
\item[1.]
Compute the initial base $\mathcal{B}_{\mathrm{init}}$ and set $\mathcal{B} = \mathcal{B}_{\mathrm{init}}$.

\item[2.]
Compute the base $\mathcal{B}'$ from $\mathcal{B}$.

\item[3.]
If  $\mathcal{B}'$ equals $\mathcal{B}$ then halt and return $\mathcal{B}$.

\item[4.]
Assign new base $\mathcal{B}'$ to $\mathcal{B}$ and go to step 2. \vspace{-1ex}
\end{itemize} \\
        \hline
\end{tabular}
\end{center}

Apparently, the algorithm relies on the base $\mathcal{B}_{\mathrm{init}}$ of the initial congruence $\equiv_0$ and the refinement step, computing $\mathcal{B}'$ from $\mathcal{B}$.

\subsection{Outline of the Algorithm}\label{subsec:naive-algorithm}

The framework of the algorithm is described as Fig.~\ref{Efficient_Algorithm}.

\subsubsection{Initial Congruence}

The base $\mathcal{B}_{\mathrm{init}} = (\mathbf{P}_{\mathrm{init}}, \mathbf{E}_{\mathrm{init}})$ of the initial congruence $\equiv_0$ is set as:
\begin{itemize}
\item
$\mathbf{P}_{\mathrm{init}} = X_1$,

\item
$\mathbf{E}_{\mathrm{init}}$ contains $X_i = \underbrace{X_1 \cdot X_1 \cdot \ldots \cdot X_1}_{\mathtt{norm}(X_i)\textrm{ times}}$ for every $i > 1$.
\end{itemize}

For  $\equiv_0$, we have the following properties.
\begin{lemma}
$\alpha \equiv_0 \beta$   if and only if $\mathtt{norm}(\alpha) = \mathtt{norm}(\beta)$.
\end{lemma}

\begin{lemma}
\begin{enumerate}
\item
${\equiv_0} \supseteq {\simeq}$.

\item
$\equiv_0$ is a norm-preserving and decompositional congruence.
\end{enumerate}
\end{lemma}

\subsubsection{Properties of Refinement Steps}
In order to understand the framework of the algorithm,
We need to investigate the relationship between  $\mathcal{B}' = (\mathbf{P}', \mathbf{E}') $ and  $\mathcal{B}  = (\mathbf{P}, \mathbf{E})$ in step~2.  Later from the algorithm, we will confirm that ${\stackrel{\mathcal{B}'}{\equiv}} \subseteq {\stackrel{\mathcal{B}}{\equiv}}$. Under this condition,  we have the following key observation.
\begin{lemma}\label{lem:PP}
Let $\mathcal{B}  = (\mathbf{P}, \mathbf{E})$  and $\mathcal{B}' = (\mathbf{P}', \mathbf{E}') $  be two decomposition bases.
\begin{enumerate}
\item
If  ${\stackrel{\mathcal{B}'}{\equiv}} \subseteq {\stackrel{\mathcal{B}}{\equiv}}$, then  $\mathbf{P} \subseteq \mathbf{P}'$.

\item
If $\mathbf{P}' = \mathbf{P}$, then $\mathcal{B}' = \mathcal{B}$.
\end{enumerate}
\end{lemma}

\begin{proof}
\begin{enumerate}
\item
Suppose $X_i \not\in \mathbf{P}'$, we show $X_i \not\in \mathbf{P}$.  Since $X_i \not\in \mathbf{P}'$, there is an equation ${X_i = \alpha}$ in  $\mathbf{E}'$ for some $\alpha \in \mathbf{C}_{i-1}^{*}$, which means $X_i \stackrel{\mathcal{B}'}{\equiv} \alpha$. Because  ${\stackrel{\mathcal{B}'}{\equiv}}  \subseteq {\stackrel{\mathcal{B}}{\equiv}}$, we have $X_i \stackrel{\mathcal{B}}{\equiv} \alpha$, which means that $X_i$ is not a prime modulo $ \stackrel{\mathcal{B}}{\equiv}$. That is, $X_i \not\in \mathbf{P}$.

\item
Suppose that $\mathcal{B}' \subsetneq \mathcal{B}$. Then there is some $X_i$ such that $\mathtt{dcmp}_{\mathcal{B}}(X_i) \neq \mathtt{dcmp}_{\mathcal{B}'}(X_i)$.  We have $ X_i \stackrel{\mathcal{B}}{\equiv} \mathtt{dcmp}_{\mathcal{B}}(X_i) $ and $ X_i \stackrel{\mathcal{B}'}{\equiv} \mathtt{dcmp}_{\mathcal{B}'}(X_i) $  Since  ${\stackrel{\mathcal{B}'}{\equiv}}  \subseteq {\stackrel{\mathcal{B}}{\equiv}}$, we have $X_i \stackrel{\mathcal{B}}{\equiv} \mathtt{dcmp}_{\mathcal{B}'}(X_i)$, thus
 $\mathtt{dcmp}_{\mathcal{B}}(X_i)  \stackrel{\mathcal{B}}{\equiv} \mathtt{dcmp}_{\mathcal{B}'}(X_i)$.
Since
$\mathtt{dcmp}_{\mathcal{B}}(X_i)$ and $\mathtt{dcmp}_{\mathcal{B}'}(X_i)$ are both in $\mathbf{P}^{*}$, we have  $\mathtt{dcmp}_{\mathcal{B}}(X_i)  = \mathtt{dcmp}_{\mathcal{B}'}(X_i)$, a contradiction. \qed
\end{enumerate}

\end{proof}

According to Lemma~\ref{lem:PP}, we call constants in $\mathbf{P}$ {\em old primes} and constants in $\mathbf{P}' \setminus \mathbf{P}$
{\em new primes}.  During the iterative procedure of refinement, once a constant becomes prime, it is a prime thereafter.
If at certain step of iteration there is no new prime to add, the algorithm terminates.  Thus  we have the following property.
\begin{proposition}
There can be at most $n$ steps of iteration in the algorithm.
\end{proposition}
This confirms the termination of the algorithm and provides an implementation of the step 3 by checking if there are new primes. The remaining thing is to study the
implementation of step~2.

\begin{figure}[tbp]

\begin{center}
 \begin{tabular}{|p{12cm}|}\hline\vspace{0ex}

\textbf{\normalsize  Framework of the algorithm:}

\begin{enumerate}
\item
Initialize $\mathcal{B} = (\mathbf{P}, \mathbf{E})$;

\item
$\mathbf{P}' \coloneqq \mathbf{P}$;

\item
\textbf{repeat}

\item
\qquad
$\mathbf{P} \coloneqq \mathbf{P}'$;
$\mathbf{E} \coloneqq \mathbf{E}'$;
$\mathbf{E}' \coloneqq \emptyset$;

\item
\qquad \textbf{for} each $X_i \in  \mathbf{C}\setminus \mathbf{P}$ \textbf{do}

\item
\qquad\qquad
$s  \coloneqq \mathtt{dcmp}_{(\mathbf{P}', \mathbf{E}')}(\alpha_i)$;

\item
\qquad\qquad $flag \coloneqq \mathbf{true}$;

\item
\qquad\qquad $k \coloneqq \mathtt{lpfindex}_{(\mathbf{P}, \mathbf{E})}(X_i)$;

\item
\qquad\qquad  \textbf{for} each $X_j \in \{ \mathtt{lpf}_{(\mathbf{P}, \mathbf{E})}(X_i) \} \cup \{X_{k+1}, \ldots, X_{i-1}\} \cap (\mathbf{P}' \setminus \mathbf{P})$ \textbf{do}

\item
\qquad\qquad\qquad
\textbf{if} $\mathtt{lpftest}_{(\mathbf{P}', \mathbf{E}')}(X_i, X_j)$
\textbf{then}

\item
\qquad\qquad\qquad\qquad
$\mathbf{E}'   \coloneqq \mathbf{E}' \cup \{ X_i = X_j  \cdot \mathtt{sffx}(\mathtt{norm}(X_i) - \mathtt{norm}(X_j); s) \}$;

\item
\qquad\qquad\qquad\qquad
$flag  \coloneqq \mathbf{false}$;

\item
\qquad\qquad\qquad
\textbf{end if}

\item
\qquad\qquad
\textbf{end for}

\item
\qquad\qquad
\textbf{if} $flag$ \textbf{then}

\item
\qquad\qquad\qquad
$\mathbf{P}' \coloneqq \mathbf{P}'\cup \{X_i\}$;

\item
\qquad\qquad
\textbf{end if}

\item
\qquad
\textbf{end for}

\item
\textbf{until} $\mathbf{P} = \mathbf{P}'$ \vspace{-1ex}
\end{enumerate}\\
        \hline
    \end{tabular}\vspace{-1ex}
\end{center}
\caption{Framework of Efficient Algorithm}\label{Efficient_Algorithm}
\end{figure}

\subsubsection{Computing $\mathcal{B}'$  from $\mathcal{B}$}

Computation of $\mathcal{B}'$ proceeds as follows. First we assign $\mathbf{P}' = \mathbf{P}$ and $\mathbf{E}' = \emptyset$. Then we add
appropriate constants to $\mathbf{P}'$ and appropriate equations to $\mathbf{E}'$. For every $i = 2, \ldots, n$ with $X_i \in \mathbf{C} \setminus \mathbf{P}$, we check whether there exists $\delta \in  (\mathbf{P}' \cap  \mathbf{C}_{i-1})^{*}$ such that $X_i  \stackrel{\mathcal{B}'}{\equiv} \delta$. If not, we
add $X_i$ to $\mathbf{P}'$, otherwise we add the appropriate equation $X_i = \alpha$ to $\mathbf{E}'$. We emphasize that at the time $X_i$ is treated, we have already known whether $X_j \in \mathbf{P}'$ and $\mathtt{dcmp}_{\mathcal{B}'} (X_j)$ for every $j<i$.

The efficient computation of  $\mathcal{B}'$  from $\mathcal{B}$ relies on the following three aspects:
\begin{enumerate}
\item
The candidates $\delta$ for testing $X_i  \stackrel{\mathcal{B}'}{\equiv} \delta$ must be `small'.

\item
We need an correct and efficient way of deciding whether $(X_i, \delta)$ can be put into $\mathbf{E}'$, i.e.~$X_i \stackrel{\mathcal{B}'}{\equiv} \delta$.

\item
We need an efficient representation and manipulation on strings.
\end{enumerate}

The representation and operations on long strings can be implemented in a systematic way  and will be discussed shortly in Section~\ref{subsec:longstring}.    For the moment, we suppose that all the operations on strings  appears in the algorithm are polynomial time computable.

\subsection{Small Set of Candidates}\label{subsec:candidates}

Now we confirm that, for every $X_i$, there is a small number of $\delta$'s which are required to determine whether $X_i \stackrel{\mathcal{B}'}{\equiv} \delta$.  In the case of realtime $\mathrm{nBPA}$, this is a significant discovery in CL algorithm, for it greatly reduces the expense of the algorithm.  The same way is taken here, but the rationality will be confirmed later.

Let $\mathcal{B} = (\mathbf{P}, \mathbf{E})$ be a decomposition base.  We say that prime constant $X_j \in \mathbf{P}$ is the
{\em leftmost prime factor} of $X_i$ wrt.~$\mathcal{B}$, denoted by $\mathtt{lpf}_{\mathcal{B}}(X_i) = X_j$, if $\mathtt{dcmp}_{\mathcal{B}}(X_i) = X_j\cdot \gamma$ for some $\gamma$. Clearly, $\mathtt{lpf}_{\mathcal{B}}(X_i)$ is unique.

Now fix one decreasing  transition rule $X_i \stackrel{\ell_i}{\longrightarrow}_{\mathrm{dec}} \alpha_i$ ($\ell_i = \tau$ is allowed. ) for every $X_i \in \mathbf{C}$.   We use $\mathtt{sffx}(h; \alpha)$ to denote the suffix of string $\alpha$ with norm $h$. Note that $\mathtt{sffx}(h; \alpha)$ is undefined unless $\alpha$ has such a suffix with norm $h$.
\begin{proposition}\label{prop:lpf1}
Let $\mathcal{B}$ be a decomposition base such that $\stackrel{\mathcal{B}}{\equiv}$ is a decreasing branching bisimulation (Definition~\ref{def:dec_bisimulation_expansion}, Section~\ref{subsec:Expansion_in_general}).  If $\mathtt{lpf}_{\mathcal{B}}(X_i) = X_j$, then
\[
X_i \stackrel{\mathcal{B}}{\equiv}  X_j \cdot \mathtt{sffx}(\mathtt{norm}(X_i)-\mathtt{norm}(X_j); \mathtt{dcmp}_{\mathcal{B}} (\alpha_i)).
\]
\end{proposition}

\begin{proof}
From $\mathtt{lpf}_{\mathcal{B}}(X_i) = X_j$, we have $X_i \stackrel{\mathcal{B}}{\equiv} X_j \cdot \alpha$ for some $\alpha$ satisfying $\mathtt{norm}(\alpha) = \mathtt{norm}(X_i)-\mathtt{norm}(X_j)$. Knowing $\stackrel{\mathcal{B}}{\equiv}$ is a decreasing branching bisimulation, we consider the transition $X_i \stackrel{\ell_i}{\longrightarrow}_{\mathrm{dec}} \alpha_i$. There are two cases:
\begin{itemize}
\item
$\ell_i = \tau$ and $\alpha_i \stackrel{\mathcal{B}}{\equiv} X_j \cdot\alpha$. In this case, let $\beta = X_j$ and we have $\alpha_i  \stackrel{\mathcal{B}}{\equiv} \beta \cdot \alpha$.

\item
$\ell_i \neq \tau$ or $\alpha_i \not\stackrel{\mathcal{B}}{\equiv} X_j \cdot\alpha$. In this case, we have $X_j \Longrightarrow_{\mathrm{dec}}  \stackrel{\ell_i}{\longrightarrow}_{\mathrm{dec}} \beta$ such that $\alpha_i \stackrel{\mathcal{B}}{\equiv} \beta\cdot\alpha$.
\end{itemize}
In either case, we have $\alpha_i \stackrel{\mathcal{B}}{\equiv} \beta\cdot\alpha$ for some $\beta$. According to the fact that $\stackrel{\mathcal{B}}{\equiv}$ is decompositional, we get $\mathtt{dcmp}_{\mathcal{B}}(\alpha_i) =  \mathtt{dcmp}_{\mathcal{B}}(\beta) \cdot \mathtt{dcmp}_{\mathcal{B}}(\alpha)$, and consequently $\mathtt{dcmp}_{\mathcal{B}}(\alpha) = \mathtt{sffx}(\mathtt{norm}(\alpha); \mathtt{dcmp}_{\mathcal{B}} (\alpha_i)) = \mathtt{sffx}(\mathtt{norm}(X_i)-\mathtt{norm}(X_j); \mathtt{dcmp}_{\mathcal{B}} (\alpha_i))$, hence $\alpha \stackrel{\mathcal{B}}{\equiv} \mathtt{sffx}(\mathtt{norm}(X_i)-\mathtt{norm}(X_j); \mathtt{dcmp}_{\mathcal{B}} (\alpha_i))$. Recall that $X_i \stackrel{\mathcal{B}}{\equiv} X_j \cdot \alpha$, we get $X_i \stackrel{\mathcal{B}}{\equiv} X_j \cdot\mathtt{sffx}(\mathtt{norm}(X_i)-\mathtt{norm}(X_j); \mathtt{dcmp}_{\mathcal{B}} (\alpha_i))$. \qed
\end{proof}

Assume that  ${\stackrel{\mathcal{B}'}{\equiv}} \subseteq  {\stackrel{\mathcal{B}}{\equiv}}$.
Comparing $\mathtt{lpf}_{\mathcal{B}'}(X_i)$ with $\mathtt{lpf}_{\mathcal{B}}(X_i)$, there are two possibilities: $\mathtt{lpf}_{\mathcal{B}'}(X_i) = \mathtt{lpf}_{\mathcal{B}}(X_i)$ or $\mathtt{lpf}_{\mathcal{B}'}(X_i) \neq \mathtt{lpf}_{\mathcal{B}}(X_i)$. If $\mathtt{lpf}_{\mathcal{B}'}(X_i) \neq \mathtt{lpf}_{\mathcal{B}}(X_i)$, the following property confirms that $\mathtt{lpf}_{\mathcal{B}'}(X_i)$ must be a new prime.
\begin{proposition}\label{prop:lpf2}
Assume that  ${\stackrel{\mathcal{B}'}{\equiv}} \subseteq  {\stackrel{\mathcal{B}}{\equiv}}$.  Let $X_{j'} = \mathtt{lpf}_{\mathcal{B}'}(X_i)$ and  $X_j = \mathtt{lpf}_{\mathcal{B}}(X_i)$. If $j' \neq j$, then $j' > j$ and $X_{j'} \in \mathbf{P}' \setminus \mathbf{P}$.
\end{proposition}

\begin{proof}
Assume on the contrary that $X_{j'} \in \mathbf{P}$, we have $X_i \stackrel{\mathcal{B}}{\equiv} X_j \cdot \gamma \stackrel{\mathcal{B}}{\equiv} X_{j'} \cdot \gamma'$, which violates unique decomposition property of $\stackrel{\mathcal{B}}{\equiv}$. \qed
\end{proof}

Now we can illustrate the algorithm framework in Fig.~\ref{Efficient_Algorithm}.  The \textbf{repeat} block at line~3 realize the procedure of iteration. At every iteration, $\mathcal{B} = (\mathbf{P}, \mathbf{E})$ is updated to $\mathcal{B}' = (\mathbf{P}', \mathbf{E}')$. During an iteration, every constant $X_i$ which is current composite is treated in the fixed index order via the outer \textbf{for} block at line~5. Note that, when $X_i$ is treated, $\mathtt{dcmp}_{\mathcal{B}'} (\alpha_i)$ can be determined. Then the inner \textbf{for} block at line~9 is used for discovering a new decomposition of $X_i$ for $\mathcal{B}'$ by determining the leftmost prime factor $X_j$ of $X_i$. By Proposition~\ref{prop:lpf2}, $X_j$ can be unchanged (in the case $X_j = \mathtt{lpf}_{(\mathbf{P}, \mathbf{E})}(X_i)$), or be a new prime less than $X_i$ (in the case $X_j  \in   (\mathbf{P}' \setminus \mathbf{P})$ and  $\mathtt{lpfindex}_{\mathcal{B}}(X_i) <  j < i$), or be $X_i$ itself (in the case no $X_j$ is found in the inner \textbf{for} block at line~9).   In the last case, variable $flag$ which is set true at line~7 remains being $\mathbf{true}$ and $X_i$ is added to the set $\mathbf{P}'$ of new primes (line~15).  The operation $\mathtt{lpfindex}_{\mathcal{B}}(X_i)$ returns index $k$ such that $\mathtt{lpf}_{\mathcal{B}}(X_i) = X_k$.
Using Proposition~\ref{prop:lpf1} and Proposition~\ref{prop:lpf2}, the set of candidates can be confined into the form of
$X_j \cdot \mathtt{sffx}(\mathtt{norm}(X_i)-\mathtt{norm}(X_j); \mathtt{dcmp}_{\mathcal{B}'} (\alpha_i))$,
Note that, in the inner \textbf{for} block,  procedure $\mathtt{lpftest}_{\mathcal{B}'}(X_i,X_j)$ is used  to check whether $X_j$ is the leftmost prime factor of $X_i$ modulo $\stackrel{\mathcal{B}'}{\equiv}$. In fact, it tests whether
\begin{equation}\label{eqn:testing_Xi_alpha}
X_i \stackrel{\mathcal{B}'}{\equiv}  X_j \cdot \mathtt{sffx}(\mathtt{norm}(X_i)-\mathtt{norm}(X_j); \mathtt{dcmp}_{\mathcal{B}'} (\alpha_i)).
\end{equation}

In the rest part of this paper,  the right hand side of Equation~(\ref{eqn:testing_Xi_alpha}) is denoted by $\delta$. We remark that $\delta \in \mathbf{P}' \cap \mathbf{C}_{i-1}$.  Our goal is to find an efficient way to check whether $X_i \stackrel{\mathcal{B}'}{\equiv} \delta$.

\begin{remark}
The small number of candidates of $\delta$ relies on Proposition~\ref{prop:lpf1}, which requires that $\stackrel{\mathcal{B}'}{\equiv}$ be a decreasing bisimulation. The definition of decreasing bisimulation will be introduced in Section~\ref{subsec:Expansion_in_general}. According to the refinement operation defined in Section~\ref{sec:refinement_steps},  $\stackrel{\mathcal{B}'}{\equiv}$ is assured to be a decreasing bisimulation.
\end{remark}

\subsection{Efficient Way of Testing $X_i \stackrel{\mathcal{B}'}{\equiv} \delta$}

The algorithm framework described in Fig.~\ref{Efficient_Algorithm} tells us an efficient way for the implementation of partition refinement on the unique decomposition congruences.  Up to now, we have not discuss how the refinement operation is and how shall we realize it efficiently. That is, how $\mathtt{lpftest}_{(\mathbf{P}', \mathbf{E}')}(X_i, X_j)$ at line~10 is implemented. Now we present the details.  That is , we present an efficient way to check whether $X_i \stackrel{\mathcal{B}'}{\equiv} \delta$. In this way, we define $\mathcal{B}'$ from $\mathcal{B}$ via the algorithm.

The whole testing is described in Fig.~\ref{Checking_EXP}. In later sections, we have further discussions on this implementation.  For now, we only remark that, in the situation of realtime $\mathrm{nBPA}$, this implementation coincides with CL algorithm.   The proof of correctness is deferred to Section~\ref{sec:refinement_steps}.

\begin{figure}[tbp]
\begin{center}
\begin{tabular}{|p{12cm}|}\hline \vspace{0ex}
\textbf{\normalsize Checking  $X_i \stackrel{\mathcal{B}'}{\equiv} \delta$:}
\begin{enumerate}
\item
\textbf{test}  $\mathtt{dcmp}_{\mathcal{B}}(X_i) = \mathtt{dcmp}_{\mathcal{B}}(\delta)$. if so, goto step~2; else \textbf{reject} $(X_i, \delta)$. \vspace{1ex}

\item
\textbf{test} for every $X_i  \stackrel{\ell}{\longrightarrow}_{\mathrm{dec}} \alpha$, we have
\begin{enumerate}
\item
    either  $\ell = \tau$ and $\mathtt{dcmp}_{\mathcal{B}'}(\alpha) = \delta$;

\item
     or $\delta  \stackrel{\ell}{\longrightarrow}_{\mathrm{dec}} \beta$ for some $\beta$ and $\mathtt{dcmp}_{\mathcal{B}'} (\alpha) = \mathtt{dcmp}_{\mathcal{B}'} (\beta)$.
\end{enumerate}
If so, goto step~3; else, \textbf{reject} $(X_i, \delta)$. \vspace{1ex}

\item
\textbf{test} for every $X_i  \stackrel{\ell}{\longrightarrow}_{\mathrm{inc}} \alpha$, we have

 $\delta  \stackrel{\ell}{\longrightarrow}_{\mathrm{inc}} \beta$ for some $\beta$ and $\mathtt{dcmp}_{\mathcal{B}} (\alpha) = \mathtt{dcmp}_{\mathcal{B}} (\beta)$.

If so, goto step~4; else, \textbf{reject} $(X_i, \delta)$. \vspace{1ex}

\item
\textbf{test} whether $X_i \stackrel{\tau}{\longrightarrow}_{\mathrm{dec}} \alpha$ for some $\alpha$ such that $\mathtt{dcmp}_{\mathcal{B}'} (\alpha) = \delta $.

If so, goto step~7; else, goto step~5. \vspace{1ex}

\item
\textbf{test}  for every $\delta  \stackrel{\ell}{\longrightarrow}_{\mathrm{dec}} \beta$, we have

    $X_i \stackrel{\ell}{\longrightarrow}_{\mathrm{dec}}  \alpha$ for some $\alpha$ such that $\mathtt{dcmp}_{\mathcal{B}'}(\alpha) = \mathtt{dcmp}_{\mathcal{B}'} (\beta)$.

If so, goto step~6; else, \textbf{reject} $(X_i, \delta)$. \vspace{1ex}

\item
\textbf{test}  for every  $\delta \stackrel{\ell}{\longrightarrow}_{\mathrm{inc}} \beta$, we have

    $X_i \stackrel{\ell}{\longrightarrow}_{\mathrm{inc}}  \alpha$ for some $\alpha$ such that $\mathtt{dcmp}_{\mathcal{B}}(\alpha) = \mathtt{dcmp}_{\mathcal{B}} (\beta)$.

If so, goto step~7; else, \textbf{reject} $(X_i, \delta)$. \vspace{1ex}

\item
\textbf{accept} $(X_i, \delta)$.
\vspace{-1ex}
\end{enumerate} \\
        \hline
\end{tabular}\vspace{-1ex}
\end{center}
\caption{Checking  $X_i \stackrel{\mathcal{B}'}{\equiv} \delta$}\label{Checking_EXP}
\end{figure}

We can state two properties which need to be used to make the whole framework Fig.~\ref{Efficient_Algorithm} work.
\begin{lemma}
In every iteration of Fig.~\ref{Efficient_Algorithm}, we get a decomposition base $\mathcal{B}'$ from $\mathcal{B}$. The following hold:
\begin{enumerate}
\item
${\stackrel{\mathcal{B}'}{\equiv}} \subseteq {\stackrel{\mathcal{B}}{\equiv}}$.

\item
${\stackrel{\mathcal{B}'}{\equiv}}$ is a decreasing bisimulation.
\end{enumerate}
\end{lemma}

\begin{proof}
 Item~1  is an inference directly from Fig.~\ref{Checking_EXP}. Item~2  will be discussed in detail in Section~\ref{sec:refinement_steps}.  \qed
\end{proof}

\subsection{Operations on Long Strings}\label{subsec:longstring}
In the algorithm, we meet quite a few operations on strings whose length is exponential. Thus we need an efficient way to represent and manipulate them.
This sort of improvement actually appears in all the previous work on strong bisimilarity checking on normed $\mathrm{BPA}$. There are many different ways to do so, and nothing special in our situation.  Thus we only sketch the idea and provide some literature.

In the previous work~\cite{DBLP:journals/tcs/HirshfeldJM96,DBLP:conf/mfcs/LasotaR06,DBLP:conf/fsttcs/CzerwinskiL10}, a long string is represented by a {\em straight-line program} (SLP), a
context-free grammar (typically in Chomsky normal form) which generates only one word.
The efficient algorithms rely on an efficient implementation of equality checking on SLP-compressed strings, which is typically implemented (as a special case) by an efficient algorithm of compressed pattern matching such as~\cite{DBLP:conf/cpm/MiyazakiST97,DBLP:conf/dagstuhl/Lifshits06}.  Lohrey~\cite{DBLP:journals/gcc/Lohrey12} gives a nice survey on algorithms on SLP-compressed strings.

One deficiency of the above scheme is that  the procedure for string equality checking is called every time two strings need to compare, and previous computations are completely ignored. In~\cite{DBLP:journals/algorithmica/MehlhornSU97} and its improved version~\cite{DBLP:conf/soda/AlstrupBR00}, a data structure for finite set of strings is maintained, which supports concatenation, splitting, and equality checking operations.  Czerwi\'{n}ski~\cite{CzerwinskiPhD} uses this technique to improve his previous algorithm~\cite{DBLP:conf/fsttcs/CzerwinskiL10}.

\subsection{Analysis of Time Complexity}
Now we give a very brief discussion of the time complexity of the whole algorithm.  Some less important factors are deliberately neglected.  Readers are  referred to~Czerwi\'{n}ski~\cite{CzerwinskiPhD}.

Consider the algorithm described in Fig.~\ref{Efficient_Algorithm}.   The dominating factor is the operation $\mathtt{lpftest}_{(\mathbf{P}', \mathbf{E}')}(X_i, X_j)$ at line~10.  We claim that there are totally $\mathcal{O}(n^2)$ invocations of $\mathtt{lpftest}$.

In the implementation of $\mathtt{lpftest}$, we call the procedures  described in Fig.~\ref{Checking_EXP}.  The procedure treats processes as {\em normed strings}.  Therefore, the time consumed depends on the costs of the operations on normed strings. We suppose that there are three operations of `normed' strings: $\mathtt{Concatenate}(\sigma_1, \sigma_2)$, $\mathtt{Split}(\sigma, h)$, and $\mathtt{Equal}(\sigma_1, \sigma_2)$, which are supposed to spend time $\mathsf{C}(N)$, $\mathsf{S}(N)$, and $\mathsf{E}(N)$, respectively.  Claimed in~\cite{CzerwinskiPhD},   the best implementation  is $\mathsf{C}(N) = \mathcal{O}(N\cdot \mathrm{polylog}N)$, $\mathsf{S}(N) = \mathcal{O}(N\cdot\mathrm{polylog}N)$, and $\mathsf{E}(N) = \mathcal{O}(\mathrm{polylog}N)$.

Consider the procedures in Fig.~\ref{Checking_EXP}.  The most time-consuming part is still the part of matching, which can perform $\mathcal{O}(N^2)$ times of $\mathtt{Equal}$ operations.   This makes the total time of checking branching bisimilarity no difference from checking strong bisimilarity. The overall running time is $\mathcal{O}(N^4\cdot\mathrm{polylog}N)$.

\section{The Refinement Operation}\label{sec:refinement_steps}

Now, we start to discuss the correctness of the algorithm.
In order to prove the correctness, we need to answer two questions:

\begin{enumerate}
\item
What is the refinement operation corresponding to a step of iteration in our algorithm?

\item
How our algorithm can be derived from the refinement operation.
 \end{enumerate}

In this section we answer the first question, and the second question will be answered in Section~\ref{sec:correctness}.

Actually how to define the refinement operation for our algorithm is really not clear at the first glance. Thus we review the refinement operation adopted in CL algorithm in Section~\ref{subsec:Expansion_Pre}. Then in Section~\ref{subsec:Expansion_Pre} we find another way to define and understand the refinement operation in Section~\ref{subsec:Expansion_Another_undestanding}.  Following this understanding, we attempt to define the refinement operation which turns out to be suitable for our algorithm in Section~\ref{subsec:Expansion_in_general}, and then show some basic properties.

\subsection{The Refinement Operation for Realtime $\mathrm{nBPA}$}\label{subsec:Expansion_Pre}
Before going into the tricky part of our definition of the refinement relation, let us review the reason why the algorithm is correct for the realtime  $\mathrm{nBPA}$. This special case is comparatively easy.  For convenience, we describe the procedure of checking $X_i \stackrel{\mathcal{B}'}{\equiv} \delta$ for realtime $\mathrm{nBPA}$ in Fig.~\ref{Checking_EXP_REALTIME}. This is nothing but a special case of~Fig.~\ref{Checking_EXP}, and it is a slightly simplified version of the corresponding procedure in CL algorithm.

\begin{figure}[tbp]
\begin{center}
\begin{tabular}{|p{12cm}|}\hline \vspace{0ex}
\textbf{\normalsize Checking  $X_i \stackrel{\mathcal{B}'}{\equiv} \delta$ in the case of REALTIME systems:}
\begin{enumerate}
\item
\textbf{test}
  $\mathtt{dcmp}_{\mathcal{B}}(X_i) = \mathtt{dcmp}_{\mathcal{B}}(\delta)$.

If so, goto step~2; else \textbf{reject} $(X_i, \delta)$. \vspace{1ex}

\item
\textbf{test} for every $X_i  \stackrel{\ell}{\longrightarrow}_{\mathrm{dec}} \alpha$, we have
\begin{quote}
$\delta  \stackrel{\ell}{\longrightarrow}_{\mathrm{dec}} \beta$ for some $\beta$ and $\mathtt{dcmp}_{\mathcal{B}'} (\alpha) = \mathtt{dcmp}_{\mathcal{B}'} (\beta)$.
\end{quote}
If so, goto step~3; else, \textbf{reject} $(X_i, \delta)$. \vspace{1ex}

\item
\textbf{test} for every $X_i  \stackrel{\ell}{\longrightarrow}_{\mathrm{inc}} \alpha$, we have
\begin{quote}
 $\delta  \stackrel{\ell}{\longrightarrow}_{\mathrm{inc}} \beta$ for some $\beta$ and $\mathtt{dcmp}_{\mathcal{B}} (\alpha) = \mathtt{dcmp}_{\mathcal{B}} (\beta)$.
\end{quote}
If so, goto step~4; else, \textbf{reject} $(X_i, \delta)$. \vspace{1ex}

\item
\textbf{test}  for every $\delta  \stackrel{\ell}{\longrightarrow}_{\mathrm{dec}} \beta$, we have
\begin{quote}
    $X_i \stackrel{\ell}{\longrightarrow}_{\mathrm{dec}}  \alpha$ for some $\alpha$ such that $\mathtt{dcmp}_{\mathcal{B}'}(\alpha) = \mathtt{dcmp}_{\mathcal{B}'} (\beta)$.
\end{quote}
If so, goto step~5; else, \textbf{reject} $(X_i, \delta)$. \vspace{1ex}

\item
\textbf{test}  for every  $\delta \stackrel{\ell}{\longrightarrow}_{\mathrm{inc}} \beta$, we have
\begin{quote}
    $X_i \stackrel{\ell}{\longrightarrow}_{\mathrm{inc}}  \alpha$ for some $\alpha$ such that $\mathtt{dcmp}_{\mathcal{B}}(\alpha) = \mathtt{dcmp}_{\mathcal{B}} (\beta)$.
\end{quote}
If so, goto step~6; else, \textbf{reject} $(X_i, \delta)$. \vspace{1ex}

\item
\textbf{accept} $(X_i, \delta)$.
\vspace{-1ex}
\end{enumerate} \\
        \hline
\end{tabular}\vspace{-1ex}
\end{center}
\caption{Checking  $X_i \stackrel{\mathcal{B}'}{\equiv} \delta$ for Realtime Systems}\label{Checking_EXP_REALTIME}
\end{figure}

At first we review the framework of the correctness proof for CL algorithm.

In the case of bisimilarity for realtime  $\mathrm{nBPA}$, we can define the following well-known {\em expansion} relation directly from the definition of bisimulation.

 \begin{definition}\label{def:sexp}
Let $\mathcal{R}$ be a binary relation on \textbf{realtime} processes.
The {\em expansion} of $\mathcal{R}$,  $\mathsf{Exp}(\mathcal{R})$, contains all pairs
$(\alpha, \beta)$ satisfying the  following conditions:

\begin{enumerate}
\item
Whenever $\alpha \stackrel{a}{\longrightarrow} \alpha'$, then $\beta  \stackrel{a}{\longrightarrow}  \beta'$ and $\alpha'\mathcal{R}\beta'$ for some $\beta'$.
\item
Whenever $\beta \stackrel{a}{\longrightarrow} \beta'$, then $\alpha  \stackrel{a}{\longrightarrow} \alpha'$ and $\alpha'\mathcal{R} \beta'$ for some $\alpha'$.
\end{enumerate}
\end{definition}

For realtime system,  a relation $\mathcal{R}$ is a bisimulation if and only if $\mathcal{R} \subseteq \mathsf{Exp}(\mathcal{R})$. Bisimilarity $\simeq$ is the largest relation $\mathcal{R}$ which  satisfies $\mathcal{R} = \mathsf{Exp}(\mathcal{R})$.

Definition~\ref{def:sexp} is well-behaved in the sense that ${\mathsf{Exp}(\equiv)}  \cap {\equiv} \subsetneq {\equiv}$ if ${\equiv}$ is not a bisimulation, and $\mathsf{Exp}(\equiv)$ is a norm-preserving congruence suppose that $\equiv$ is. However, we cannot simply define the refinement relation $\mathsf{Ref}(\equiv)$ to be ${\mathsf{Exp}(\equiv)} \cap {\equiv}$, because $\mathsf{Exp}(\equiv) \cap {\equiv}$ may not be a decompositional congruence even if $\equiv$ is.  In other words, we cannot always find a $\mathcal{B}'$ such that ${\stackrel{\mathcal{B}'}{\equiv}} = {\mathsf{Exp}(\stackrel{\mathcal{B}}{\equiv})}  \cap {\stackrel{\mathcal{B}}{\equiv}}$.  The way to solve this problem is to find a decompositional congruence ${\stackrel{\mathcal{B}'}{\equiv}} = \mathsf{Ref}(\stackrel{\mathcal{B}}{\equiv})$  which lies between $\simeq$ and ${\mathsf{Exp}(\stackrel{\mathcal{B}}{\equiv})}  \cap {\stackrel{\mathcal{B}}{\equiv}}$. The way suggested in~\cite{DBLP:journals/mscs/HirshfeldJM96} is that  $ \mathsf{Ref}(\equiv)$ be the decreasing bisimilarity wrt.~${\mathsf{Exp}(\equiv)}   \cap {\equiv}$.

\begin{definition}\label{def:relative_realtime}
Let $\mathcal{R}$ be a relation on \textbf{realtime} processes. $\mathcal{R}$ is a {\em decreasing bisimulation}  if the following hold whenever $\alpha \mathcal{R} \beta$:
\begin{enumerate}
\item
Whenever $\alpha \stackrel{\ell}{\longrightarrow}_{\mathrm{dec}} \alpha'$, then
$\beta \stackrel{\ell}{\longrightarrow}_{\mathrm{dec}} \beta'$ such that $\alpha'\mathcal{R} \beta'$.

\item
Whenever $\beta \stackrel{\ell}{\longrightarrow}_{\mathrm{dec}} \beta'$, then
$\alpha \stackrel{\ell}{\longrightarrow}_{\mathrm{dec}} \alpha'$ such that $\alpha'\mathcal{R} \beta'$.
\end{enumerate}

Let $\equiv$ be a norm-preserving congruence.  The {\em decreasing bisimilarity wrt. $\equiv$}, denoted by  $\simeq_{\mathrm{dec}}^{\equiv}$, is the largest decreasing bisimulation contained in $\equiv$.
\end{definition}

We do not justify the rationality of the relation $\simeq_{\mathrm{dec}}^{\equiv}$. The fact is  that $\simeq_{\mathrm{dec}}^{\equiv}$ is a congruence, and moreover, it satisfies the following:
\begin{enumerate}
\item
$\simeq_{\mathrm{dec}}^{\equiv}$ is decompositional if $\equiv$ is right-cancellative.

\item
$\mathsf{Exp}(\equiv)$ and also ${\mathsf{Exp}(\equiv)} \cap {\equiv}$ is right-cancellative if ${\equiv}$ is decompositional.
\end{enumerate}

According to these two facts, $\simeq_{\mathrm{dec}}^{{\mathsf{Exp}(\equiv)} \cap {\equiv}}$ is decompositional whenever ${\equiv}$ is.  From here, we can define $\mathsf{Ref}(\equiv)$ to be $\simeq_{\mathrm{dec}}^{{\mathsf{Exp}(\equiv)} \cap {\equiv}}$.

In order to get a characterization of $\simeq_{\mathrm{dec}}^{{\mathsf{Exp}(\equiv)} \cap {\equiv}}$, we need the following
 expansion relation for decreasing bisimilarity.

\begin{definition}
Let $\mathcal{R}$ be a binary relation on \textbf{realtime} processes. The {\em decreasing expansion} of $\mathcal{R}$,  $\mathsf{Exp}_{\mathrm{dec}}(\mathcal{R})$, contains all pairs
$(\alpha, \beta)$ satisfying the following conditions:
\begin{enumerate}
\item
Whenever $\alpha \stackrel{\ell}{\longrightarrow}_{\mathrm{dec}} \alpha'$, then
  $\beta \stackrel{\ell}{\longrightarrow}_{\mathrm{dec}} \beta'$ and $\alpha'\mathcal{R} \beta'$  for some $\beta'$.

\item
Whenever  $\beta \stackrel{\ell}{\longrightarrow}_{\mathrm{dec}} \beta'$, then
  $\alpha \stackrel{\ell}{\longrightarrow}_{\mathrm{dec}} \alpha'$ and $\alpha'\mathcal{R} \beta'$  for some $\alpha'$.
\end{enumerate}
\end{definition}
Then
we can establish the following important property for realtime $\mathrm{nBPA}$:
$(\alpha, \beta) \in \simeq_{\mathrm{dec}}^{\equiv}$ if and only if
\begin{center}
$\alpha \equiv \beta$ and $(\alpha, \beta) \in \mathsf{Exp}_{\mathrm{dec}}({\simeq_{\mathrm{dec}}^{\equiv}})$.
\end{center}
From this fact, considering that ${\mathsf{Ref}(\equiv)} = {\simeq_{\mathrm{dec}}^{{\mathsf{Exp}(\equiv)}\cap {\equiv}}}$, we have:
$(\alpha, \beta) \in  \mathsf{Ref}(\equiv)$ if and only if
\begin{center}
$\alpha \equiv \beta$ and
$(\alpha,\beta) \in {\mathsf{Exp}(\equiv)}$  and $(\alpha, \beta) \in \mathsf{Exp}_{\mathrm{dec}}(\mathsf{Ref}(\equiv))$.
\end{center}
Now, to prove ${\stackrel{\mathcal{B}'}{\equiv}} = {\mathsf{Ref}(\stackrel{\mathcal{B}}{\equiv})}$, it suffices to prove:
\begin{center}
$\alpha \stackrel{\mathcal{B}'}{\equiv} \beta$ if and only if $\alpha \stackrel{\mathcal{B}}{\equiv} \beta$ and $(\alpha,\beta) \in \mathsf{Exp}(\stackrel{\mathcal{B}}{\equiv})$ and $(\alpha, \beta) \in \mathsf{Exp}_{\mathrm{dec}}(\stackrel{\mathcal{B}'}{\equiv})$.
\end{center}

According to this characterization, apparently we have ${\stackrel{\mathcal{B}'}{\equiv}} \subseteq {\stackrel{\mathcal{B}}{\equiv}}$.

Now it is time to  explain that
 the procedure in Fig.~\ref{Checking_EXP_REALTIME} is actually based on this characterization.
Suppose we want to check whether $X_i \stackrel{\mathcal{B}'}{\equiv} \delta$.  It suffices to check the following three conditions:
\begin{enumerate}
\item
$X_i \stackrel{\mathcal{B}}{\equiv} \delta$.

\item
$(X_i, \delta) \in \mathsf{Exp}_{\mathrm{dec}}(\stackrel{\mathcal{B}'}{\equiv})$.

\item
$(X_i,\delta) \in \mathsf{Exp}(\stackrel{\mathcal{B}}{\equiv})$.
\end{enumerate}

Notice that these three conditions are deliberately arranged in the above order.
Now we study the procedure described in~Fig.~\ref{Checking_EXP_REALTIME}. Step~1 corresponds to Condition~1: checking $X_i \stackrel{\mathcal{B}}{\equiv} \delta$. Step~2 and Step~4 correspond to Condition~2: checking $(X_i, \delta) \in \mathsf{Exp}_{\mathrm{dec}}(\stackrel{\mathcal{B}'}{\equiv})$. Step~3 and Step~5 partly  correspond to Condition~3: checking $(X_i,\delta) \in \mathsf{Exp}(\stackrel{\mathcal{B}}{\equiv})$. In  Step~3 and Step~5, we find that only increasing transitions are treated. This is because the decreasing transitions are already treated in Step~2 and Step~4, in which stricter requirements are tested, considering  ${\stackrel{\mathcal{B}'}{\equiv}} \subseteq {\stackrel{\mathcal{B}}{\equiv}}$.

\subsection{Another Understanding of the Refinement Operation}\label{subsec:Expansion_Another_undestanding}

The characterization of the refinement operation defined in Section~\ref{subsec:Expansion_Pre} is fine. However, currently we do not know how to generalize this characterization to non-realtime systems.  The main problem is that we cannot find a feasible way to define the expansion relation. This is because the technique of dynamic programming is used in the algorithm. This makes the expansion of $\stackrel{\mathcal{B}}{\equiv}$, if there is a way to define, not only depend on $\stackrel{\mathcal{B}}{\equiv}$, but also depend on $\stackrel{\mathcal{B}'}{\equiv}$. This fact makes it very difficult to generalize the correctness proof in the way taken in CL algorithm.  Thus we hope to find another better way to prove the correctness of our algorithm.

Before doing this in non-realtime systems,  the attempt is first made in  realtime systems.  That is, we develop another characterization of the refinement operation for the procedure in Fig.~\ref{Checking_EXP_REALTIME}.

The basic idea is to integrate the three parts into a whole concept, which we called {\em decreasing bisimilarity with expansion}.

To avoid confusion, readers are suggested to forget the terminologies and notations taken in Section~\ref{subsec:Expansion_Pre}, because the forms of the following terminologies and  notations can be close to the ones in Section~\ref{subsec:Expansion_Pre}, but their meanings are different.

 We do not provide proofs for the lemmas  and theorems below, because they are special cases for those in Section~\ref{subsec:Expansion_in_general}.

\begin{definition}\label{def:dec_bisimulation_expansion_realtime}
Let $\equiv$ be a norm-preserving congruence on \textbf{realtime} processes, and
let ${\mathcal{R}} \subseteq {\equiv}$ be a relation on \textbf{realtime} processes.  We say $\mathcal{R}$ is a {\em decreasing bisimulation with expansion of} $\equiv$  if the following conditions hold whenever $\alpha \mathcal{R} \beta$:
\begin{enumerate}
\item
Whenever $\alpha \stackrel{\ell}{\longrightarrow}  \alpha'$,
\begin{enumerate}
\item
if  $\alpha \stackrel{\ell}{\longrightarrow}_{\mathrm{dec}}  \alpha'$,
then
$\beta \stackrel{\ell}{\longrightarrow}_{\mathrm{dec}} \beta'$ for some $\beta'$ such that $\alpha'\mathcal{R} \beta'$;

\item
if  $\alpha \stackrel{\ell}{\longrightarrow}_{\mathrm{inc}}  \alpha'$,
then
$\beta \stackrel{\ell}{\longrightarrow}_{\mathrm{inc}} \beta'$ for some $\beta'$ such that $\alpha' \equiv \beta'$.
\end{enumerate}

\item
Whenever $\beta \stackrel{\ell}{\longrightarrow} \beta'$,
\begin{enumerate}
\item
if $\beta \stackrel{\ell}{\longrightarrow}_{\mathrm{dec}} \beta'$,
then
$\alpha \stackrel{\ell}{\longrightarrow}_{\mathrm{dec}} \alpha'$ for some $\alpha'$ such that $\alpha'\mathcal{R} \beta'$.

\item
if $\beta \stackrel{\ell}{\longrightarrow}_{\mathrm{inc}} \beta'$,
then
$\alpha \stackrel{\ell}{\longrightarrow}_{\mathrm{inc}} \alpha'$ for some $\alpha'$ such that $\alpha'\equiv \beta'$.
\end{enumerate}
\end{enumerate}

The {\em decreasing bisimilarity with expansion of $\equiv$}, denoted by  $\simeq^{\equiv}$, is the largest decreasing bisimulation with expansion of ${\equiv}$.
\end{definition}

The following lemma confirms the validity of Definition~\ref{def:dec_bisimulation_expansion_realtime}.

\begin{lemma}\label{lem:property_decreasing_bisimulation_expansion_realtime}
The following properties hold:
\begin{enumerate}
\item
The identity relation is a decreasing bisimulation with expansion of $\equiv$.

\item
Let $\mathcal{R}$ be a decreasing bisimulation with expansion of $\equiv$. Then, $\mathcal{R}^{-1}$ is also a decreasing  bisimulation with expansion of $\equiv$.

\item
Let $\mathcal{R}_1$ and $\mathcal{R}_2$ be two decreasing  bisimulation with expansion of $\equiv$. Then, $\mathcal{R}_1 \circ \mathcal{R}_2$  is also a decreasing  bisimulation with expansion of $\equiv$.

\item
Let $\{\mathcal{R}_{\lambda}\}_{\lambda \in I}$ be a set of  decreasing  bisimulation with expansion of $\equiv$. Then, $\bigcup_{\lambda \in I} \mathcal{R}_{\lambda}$ is a decreasing  bisimulation with expansion of $\equiv$.
\end{enumerate}
\end{lemma}

According to Lemma~\ref{lem:property_decreasing_bisimulation_expansion_realtime}, $\simeq^{\equiv}$ is an equivalence relation. According to Definition~\ref{def:dec_bisimulation_expansion_realtime}, any decreasing  bisimulation with expansion of $\equiv$ must be norm-preserving, thus $\simeq^{\equiv}$ is also norm-preserving.  Moreover, we have

\begin{lemma}
$\simeq^{\equiv}$ is a norm-preserving congruence.
\end{lemma}

Now we can define $\mathsf{Ref}(\equiv) = {\simeq^{\equiv}}$. The validity depends on the following two properties.
\begin{lemma}
\begin{enumerate}
\item
$ {\simeq^{\simeq}} = {\simeq}$.

\item
If ${\simeq} \subsetneq {\equiv}$, then $ {\simeq} \subseteq {\simeq^{\equiv}} \subsetneq {\equiv}$.
\end{enumerate}
\end{lemma}

The unique decomposition property of $\simeq^{\equiv}$ can be established in the same way as that of $\simeq$, but relies on the right cancellation property of ${\equiv}$.

\begin{theorem}[Unique Decomposition Property of $\simeq^{\equiv}$]\label{thm:unique-decomposition-relative-realtime}
Let ${\equiv}$ be a norm-preserving congruence which is right-cancellative. Then,
$\simeq^{\equiv}$ is decompositional.
\end{theorem}

It is not hard to establish the following characterization theorem of $\simeq^{\equiv}$.

\begin{theorem}\label{lem:char_relative_bisimilarity_realtime}
Let $\alpha$, $\beta$ be realtime nBPA processes.  Then,
$\alpha \simeq^{\equiv} \beta$ if and only if $\alpha \equiv \beta$ and \begin{enumerate}
\item
Whenever $\alpha \stackrel{\ell}{\longrightarrow}  \alpha'$,
\begin{enumerate}
\item
if  $\alpha \stackrel{\ell}{\longrightarrow}_{\mathrm{dec}}  \alpha'$,
then
$\beta \stackrel{\ell}{\longrightarrow}_{\mathrm{dec}} \beta'$ for some $\beta'$ such that $\alpha' \simeq^{\equiv} \beta'$;

\item
if  $\alpha \stackrel{\ell}{\longrightarrow}_{\mathrm{inc}}  \alpha'$,
then
$\beta \stackrel{\ell}{\longrightarrow}_{\mathrm{inc}} \beta'$ for some $\beta'$ such that $\alpha' \equiv \beta'$.
\end{enumerate}

\item
Whenever $\beta \stackrel{\ell}{\longrightarrow} \beta'$,
\begin{enumerate}
\item
if $\beta \stackrel{\ell}{\longrightarrow}_{\mathrm{dec}} \beta'$,
then
$\alpha \stackrel{\ell}{\longrightarrow}_{\mathrm{dec}} \alpha'$ for some $\alpha'$ such that $\alpha' \simeq^{\equiv}\beta'$.

\item
if $\beta \stackrel{\ell}{\longrightarrow}_{\mathrm{inc}} \beta'$,
then
$\alpha \stackrel{\ell}{\longrightarrow}_{\mathrm{inc}} \alpha'$ for some $\alpha'$ such that $\alpha'\equiv \beta'$.
\end{enumerate}
\end{enumerate}
\end{theorem}
When $\mathsf{Ref}(\equiv)$ is defined as $\simeq^{\equiv}$, namely ${\stackrel{\mathcal{B}'}{\equiv}} = {\simeq^{\stackrel{\mathcal{B}}{\equiv}}}$, using Theorem~\ref{lem:char_relative_bisimilarity_realtime}, we can get exactly the procedure of checking $X_i \stackrel{\mathcal{B}'}{\equiv}  \delta$ in Fig.~\ref{Checking_EXP_REALTIME}.

It should be stressed that $\simeq^{\equiv}$ defined in Definition~\ref{def:dec_bisimulation_expansion_realtime} is exact the same as  according to ${\simeq_{\mathrm{dec}}^{{\mathsf{Exp}(\equiv)}\cap {\equiv}}}$ in Section~\ref{subsec:Expansion_Pre}.  They are two different understandings of the same refinement operation.

\subsection{The Refinement Operation for Non-realtime Systems}\label{subsec:Expansion_in_general}

We spend a lot of space to discuss the correctness of the algorithm for realtime processes. The reason is that we want to generalize the way to show the correctness of our algorithm of checking branching bisimilarity for totally normed BPA. It turns out that the classical proof for CL algorithm cannot be generalized directly.  So we find another characterization of the refinement operation in Section~\ref{subsec:Expansion_Another_undestanding}.  It turns out that this one, as expected, can be used to show the correctness of our algorithm described in Fig.~\ref{Checking_EXP}.  In this section  we discuss the refinement operation in detail.

We start from the notion of {\em decreasing bisimilarity with expansion}.

\begin{definition}\label{def:dec_bisimulation_expansion}
Let $\equiv$ be a norm-preserving congruence on processes, and
let ${\mathcal{R}} \subseteq {\equiv}$ be a relation on processes.  We say $\mathcal{R}$ is a {\em decreasing bisimulation with expansion of} $\equiv$  if the following conditions hold whenever $\alpha \mathcal{R} \beta$:

\begin{enumerate}
\item
Whenever $\alpha \stackrel{\ell}{\longrightarrow} \alpha'$, then
\begin{enumerate}
\item
either $\ell = \tau$ and $\beta  \stackrel{\tau}{\longrightarrow}_{\mathrm{dec}} \beta^1 \stackrel{\tau}{\longrightarrow}_{\mathrm{dec}} \ldots \stackrel{\tau}{\longrightarrow}_{\mathrm{dec}} \beta^i \stackrel{\tau}{\longrightarrow}_{\mathrm{dec}} \ldots \stackrel{\tau}{\longrightarrow}_{\mathrm{dec}} \beta^m$ for some $m \geq 0$ and $\beta^1,\ldots,\beta^m$ such that $\alpha'\mathcal{R} \beta^m$ and $\alpha \mathcal{R} \beta^k$ for every $1 \leq i \leq m$;

\item
or $\beta  \stackrel{\tau}{\longrightarrow}_{\mathrm{dec}} \beta^1 \stackrel{\tau}{\longrightarrow}_{\mathrm{dec}} \ldots \stackrel{\tau}{\longrightarrow}_{\mathrm{dec}} \beta^i \stackrel{\tau}{\longrightarrow}_{\mathrm{dec}} \ldots \stackrel{\tau}{\longrightarrow}_{\mathrm{dec}} \beta^m \stackrel{\ell}{\longrightarrow}_{\mathrm{dec}} \beta'$ for some $m \geq 0$ and $\beta^1,\ldots,\beta^m$ and $\beta'$ such that $\alpha'\mathcal{R} \beta'$ and $\alpha \mathcal{R} \beta^k$ for every $1 \leq i \leq m$.

\item
or $\beta  \stackrel{\tau}{\longrightarrow}_{\mathrm{dec}} \beta^1 \stackrel{\tau}{\longrightarrow}_{\mathrm{dec}} \ldots \stackrel{\tau}{\longrightarrow}_{\mathrm{dec}} \beta^i \stackrel{\tau}{\longrightarrow}_{\mathrm{dec}} \ldots \stackrel{\tau}{\longrightarrow}_{\mathrm{dec}} \beta^m \stackrel{\ell}{\longrightarrow}_{\mathrm{inc}} \beta'$ for some $m \geq 0$ and $\beta^1,\ldots,\beta^m$ and $\beta'$ such that $\alpha' \equiv \beta'$ and $\alpha \mathcal{R} \beta^k$ for every $1 \leq i \leq m$.
\end{enumerate}

\item
Whenever $\beta \stackrel{\ell}{\longrightarrow}  \beta'$, then
\begin{enumerate}
\item
either $\ell=\tau$ and    $\alpha  \stackrel{\tau}{\longrightarrow}_{\mathrm{dec}} \alpha^1 \stackrel{\tau}{\longrightarrow}_{\mathrm{dec}} \ldots  \stackrel{\tau}{\longrightarrow}_{\mathrm{dec}} \alpha^i \stackrel{\tau}{\longrightarrow}_{\mathrm{dec}} \ldots \stackrel{\tau}{\longrightarrow}_{\mathrm{dec}}  \alpha^m$ for some $m \geq 0$ and $\alpha^1,\ldots,\alpha^m$ such that $\alpha^m\mathcal{R} \beta'$ and $\alpha^k \mathcal{R} \beta$ for every $1 \leq i \leq m$;

\item
or $\alpha  \stackrel{\tau}{\longrightarrow}_{\mathrm{dec}} \alpha^1 \stackrel{\tau}{\longrightarrow}_{\mathrm{dec}} \ldots  \stackrel{\tau}{\longrightarrow}_{\mathrm{dec}} \alpha^i \stackrel{\tau}{\longrightarrow}_{\mathrm{dec}} \ldots \stackrel{\tau}{\longrightarrow}_{\mathrm{dec}}  \alpha^m \stackrel{\ell}{\longrightarrow}_{\mathrm{dec}} \alpha'$ for some $m \geq 0$ and $\alpha^1,\ldots,\alpha^m$ and $\alpha'$ such that $\alpha'\mathcal{R} \beta'$ and $\alpha^k \mathcal{R} \beta$ for every $1 \leq i \leq m$.

\item
or $\alpha  \stackrel{\tau}{\longrightarrow}_{\mathrm{dec}} \alpha^1 \stackrel{\tau}{\longrightarrow}_{\mathrm{dec}} \ldots  \stackrel{\tau}{\longrightarrow}_{\mathrm{dec}} \alpha^i \stackrel{\tau}{\longrightarrow}_{\mathrm{dec}} \ldots \stackrel{\tau}{\longrightarrow}_{\mathrm{dec}}  \alpha^m \stackrel{\ell}{\longrightarrow}_{\mathrm{inc}} \alpha'$ for some $m \geq 0$ and $\alpha^1,\ldots,\alpha^m$ and $\alpha'$ such that $\alpha'\equiv \beta'$ and $\alpha^k \mathcal{R} \beta$ for every $1 \leq i \leq m$.
\end{enumerate}
\end{enumerate}

The {\em decreasing bisimilarity with expansion of $\equiv$}, denoted by  $\simeq^{\equiv}$, is the largest decreasing bisimulation with expansion of ${\equiv}$.

If a relation ${\mathcal{R}} \subseteq {\equiv}$ satisfies the above conditions except for 1(c) and 2(c) whenever $\alpha \,\mathcal{R}\,\beta$, then we call $\mathcal{R}$ a {\em decreasing bisimulation}.
\end{definition}

Some explanation should be made on Definition~\ref{def:dec_bisimulation_expansion}.

Firstly, assume that $\alpha  \stackrel{\ell}{\longrightarrow} \alpha'$ for instance.
We know the corresponding transition is increasing or decreasing.  If the transition is increasing, the only possibility is to take the matched transitions as the item~1(c).  If the transition is decreasing, there are two subcases. The item~1(a) corresponds to the situation of $\ell = \tau$ and this silent transition can be vacantly matched.  The item~1(b) corresponds to the situation that either $\ell$ is not silent, or the silent transition must be explicitly matched.  Whenever $\alpha  \stackrel{\tau}{\longrightarrow}_{\mathrm{dec}} \alpha'$, we cannot tell which one of item~1(a) or item~1(b) should be chosen. So we must test the condition~1(a), and if 1(a) does not hold then we test condition~1(b).

Secondly, when a transition $\alpha  \stackrel{\ell}{\longrightarrow} \alpha'$ is matched by $\beta$, Definition~\ref{def:dec_bisimulation_expansion} takes a different style from Definition~\ref{def:beq}, the common definition of branching bisimulation.  Consider the condition~1(b) for example. In this case we require that the matching sequence of $\beta$ to be
\[
\beta  \stackrel{\tau}{\longrightarrow}_{\mathrm{dec}} \beta^1 \stackrel{\tau}{\longrightarrow}_{\mathrm{dec}} \ldots \stackrel{\tau}{\longrightarrow}_{\mathrm{dec}} \beta^i \stackrel{\tau}{\longrightarrow}_{\mathrm{dec}} \ldots \stackrel{\tau}{\longrightarrow}_{\mathrm{dec}} \beta^m \stackrel{\ell}{\longrightarrow}_{\mathrm{dec}} \beta'
\]
such that  $\alpha'\mathcal{R} \beta'$ and $\alpha \mathcal{R} \beta^i$ for every $1 \leq i \leq m$.  That is , every intermediate $\beta^i$ must be related to $\alpha$.  In Definition~\ref{def:beq}, however, we take the simplified matching sequence of $\beta$:
\[
\beta  \stackrel{\tau}{\Longrightarrow}_{\mathrm{dec}} \beta'' \stackrel{\ell}{\longrightarrow}_{\mathrm{dec}} \beta'
\]
such that  $\alpha'\mathcal{R} \beta'$ and $\alpha \mathcal{R} \beta''$.  The reason is explained as follows.  In the normal definition of branching bisimulation, although we do not require $\alpha \mathcal{R} \beta^i$ for every intermediate $\beta^i$, the largest bisimulation, $\simeq$, satisfy the Computation Lemma (Lemma~\ref{computation-lemma}). Thus if $\mathcal{R}$ is replaced by $\simeq$, namely if $\alpha \simeq \beta$ and $\alpha  \stackrel{\ell}{\longrightarrow} \alpha'$ is matched by
\[
\beta  \stackrel{\tau}{\longrightarrow}_{\mathrm{dec}} \beta^1 \stackrel{\tau}{\longrightarrow}_{\mathrm{dec}} \ldots \stackrel{\tau}{\longrightarrow}_{\mathrm{dec}} \beta^i \stackrel{\tau}{\longrightarrow}_{\mathrm{dec}} \ldots \stackrel{\tau}{\longrightarrow}_{\mathrm{dec}} \beta^m \stackrel{\ell}{\longrightarrow}_{\mathrm{dec}} \beta'
\]
such that  $\alpha' \simeq \beta'$ and $\alpha \simeq  \beta^m$, then we immediately have $\alpha \simeq  \beta^i$ for every $1 \leq i \leq m$.  But at present we cannot establish Computation Lemma for $\simeq^{\equiv}$, since this property depends on another equivalence $\equiv$, and in Definition~\ref{def:dec_bisimulation_expansion} we do not impose any restrictions on $\equiv$. Thus Computation Lemma could not be established if normal style matchings are taken in Definition~\ref{def:dec_bisimulation_expansion}.  Thus one way of defining decreasing bisimulation with expansion of $\equiv$ is to strengthen the relevant requirements.  We take this style not because we need the Computation Lemma, but because we need the conditions appearing in Definition~\ref{def:dec_bisimulation_expansion} to be close to the conditions checked by the algorithm.

Thirdly, the `semi-branching' style (see Definition~\ref{def:semi-beq}) is taken in the case of vacant matching. This is not necessary but is helpful to show the transitivity of $\simeq^{\equiv}$.

The following lemma confirms that the relation $\simeq^{\equiv}$ is well-defined.

\begin{lemma}\label{lem:property_decreasing_bisimulation_expansion}
The following properties hold:
\begin{enumerate}
\item
The identity relation is a decreasing bisimulation with expansion of $\equiv$.

\item
Let $\mathcal{R}$ be a decreasing bisimulation with expansion of $\equiv$. Then, $\mathcal{R}^{-1}$ is also a decreasing  bisimulation with expansion of $\equiv$.

\item
Let $\mathcal{R}_1$ and $\mathcal{R}_2$ be two decreasing  bisimulation with expansion of $\equiv$. Then, $\mathcal{R}_1 \circ \mathcal{R}_2$  is also a decreasing  bisimulation with expansion of $\equiv$.

\item
Let $\{\mathcal{R}_{\lambda}\}_{\lambda \in I}$ be a set of  decreasing  bisimulation with expansion of $\equiv$. Then, $\bigcup_{\lambda \in I} \mathcal{R}_{\lambda}$ is a decreasing  bisimulation with expansion of $\equiv$.
\end{enumerate}
\end{lemma}

According to Lemma~\ref{lem:property_decreasing_bisimulation_expansion}, $\simeq^{\equiv}$ is an equivalence relation.

Since $\simeq^{\equiv}$ is a decreasing  bisimulation with expansion of $\equiv$ according to Definition~\ref{def:dec_bisimulation_expansion}, we have
\begin{lemma}\label{lem:decreasing_bisimulation_inequality}
${\simeq^{\equiv}} \subseteq {\equiv}$.
\end{lemma}

According to Definition~\ref{def:dec_bisimulation_expansion}, any decreasing  bisimulation with expansion of $\equiv$ must be norm-preserving, thus $\simeq^{\equiv}$ is also norm-preserving.  Moreover, $\simeq^{\equiv}$ is a congruence.

\begin{lemma}
$\simeq^{\equiv}$ is a norm-preserving congruence.
\end{lemma}

\begin{proof}
We only show that $\simeq^{\equiv}$ is a congruence.
Let
\[
\mathcal{S} = \{(\alpha_1 \cdot \alpha_2,  \beta_1\cdot \beta_2) \;|\;  \alpha_1 \simeq^{\equiv} \beta_1  \mbox{ and } \alpha_2 \simeq^{\equiv} \beta_2\}  \cup {\simeq^{\equiv}}.
\]
We show $\mathcal{S}$ is a decreasing  bisimulation with expansion of $\equiv$. This is done by
checking the conditions in Definition~\ref{def:dec_bisimulation_expansion}
for every $(\alpha_1 \cdot \alpha_2,  \beta_1\cdot \beta_2) \in \mathcal{S}$.

If $\alpha_1  = \epsilon = \beta_1$.  This is a trivial case.

If $\alpha_1 \neq \epsilon$ and $\beta_1 \neq \epsilon$. This proof is done by case studies.
We study only two cases. Other cases are similar.

\begin{itemize}
 \item
 Suppose there is a transition $\alpha_1\alpha_2 \stackrel{\ell}{\longrightarrow}_{\mathrm{dec}} \alpha_1'\alpha_2$, we shall find the matching from $\beta_1\beta_2$. Remember  $\alpha_1 \simeq^{\equiv} \beta_1$, thus every transition $\alpha_1 \stackrel{\ell}{\longrightarrow}_{\mathrm{dec}} \alpha_1'$ has a matching from $\beta_1$. Say, we have the matching:
\[
\beta_1  \stackrel{\tau}{\longrightarrow}_{\mathrm{dec}} \beta_1^1 \stackrel{\tau}{\longrightarrow}_{\mathrm{dec}} \ldots \stackrel{\tau}{\longrightarrow}_{\mathrm{dec}} \beta_1^i \stackrel{\tau}{\longrightarrow}_{\mathrm{dec}} \ldots \stackrel{\tau}{\longrightarrow}_{\mathrm{dec}} \beta_1^m \stackrel{\ell}{\longrightarrow}_{\mathrm{dec}} \beta_1'
\]
such that  $\alpha_1' \simeq^{\equiv} \beta_1'$ and $\alpha_1 \simeq^{\equiv} \beta_1^i$ for every $1 \leq i \leq m$.  Then we have
\[
\beta_1\beta_2  \stackrel{\tau}{\longrightarrow}_{\mathrm{dec}} \beta_1^1\beta_2 \stackrel{\tau}{\longrightarrow}_{\mathrm{dec}} \ldots \stackrel{\tau}{\longrightarrow}_{\mathrm{dec}} \beta_1^i\beta_2 \stackrel{\tau}{\longrightarrow}_{\mathrm{dec}} \ldots \stackrel{\tau}{\longrightarrow}_{\mathrm{dec}} \beta_1^m\beta_2 \stackrel{\ell}{\longrightarrow}_{\mathrm{dec}} \beta_1'\beta_2.
\]
According to the definition of $ \mathcal{S}$ and the fact $\alpha_2 \simeq^{\equiv} \beta_2$,  we have $(\alpha_1'\alpha_2,  \beta_1'\beta_2) \in \mathcal{S}$, and $(\alpha_1\alpha_2, \beta_1^i\beta_2) \in  \mathcal{S}$ for every $1 \leq i \leq m$.

\item
 Suppose there is a transition $\alpha_1\alpha_2 \stackrel{\ell}{\longrightarrow}_{\mathrm{inc}} \alpha_1'\alpha_2$, we shall find the matching from $\beta_1\beta_2$. Remember  $\alpha_1 \simeq^{\equiv} \beta_1$, thus every transition $\alpha_1 \stackrel{\ell}{\longrightarrow}_{\mathrm{inc}} \alpha_1'$ has a matching from $\beta_1$. Say, we have the matching:
\[
\beta_1  \stackrel{\tau}{\longrightarrow}_{\mathrm{dec}} \beta_1^1 \stackrel{\tau}{\longrightarrow}_{\mathrm{dec}} \ldots \stackrel{\tau}{\longrightarrow}_{\mathrm{dec}} \beta_1^i \stackrel{\tau}{\longrightarrow}_{\mathrm{dec}} \ldots \stackrel{\tau}{\longrightarrow}_{\mathrm{dec}} \beta_1^m \stackrel{\ell}{\longrightarrow}_{\mathrm{inc}} \beta_1'
\]
such that  $\alpha_1'  \equiv \beta_1'$ and $\alpha_1 \simeq^{\equiv} \beta_1^i$ for every $1 \leq i \leq m$.  Then we have
\[
\beta_1\beta_2  \stackrel{\tau}{\longrightarrow}_{\mathrm{dec}} \beta_1^1\beta_2 \stackrel{\tau}{\longrightarrow}_{\mathrm{dec}} \ldots \stackrel{\tau}{\longrightarrow}_{\mathrm{dec}} \beta_1^i\beta_2 \stackrel{\tau}{\longrightarrow}_{\mathrm{dec}} \ldots \stackrel{\tau}{\longrightarrow}_{\mathrm{dec}} \beta_1^m\beta_2 \stackrel{\ell}{\longrightarrow}_{\mathrm{inc}} \beta_1'\beta_2.
\]
According to the definition of $ \mathcal{S}$ and the fact $\alpha_2 \simeq^{\equiv} \beta_2$,  we have  $(\alpha_1\alpha_2, \beta_1^i\beta_2) \in  \mathcal{S}$ for every $1 \leq i \leq m$.   Knowing ${\simeq^{\equiv}} \subseteq {\equiv}$, we have $\alpha_2 \equiv \beta_2$, and by congruence of ${\equiv}$ we have $\alpha_1'\alpha_2 \equiv  \beta_1'\beta_2$. \qed

\end{itemize}
\end{proof}

Now we can define the refinement operation $\mathsf{Ref}(\equiv)$ as ${\simeq^{\equiv}}$. The validity of this definition depends on the following lemma.
\begin{lemma}
The following two properties hold.
\begin{enumerate}
\item
$ {\simeq^{\simeq}} = {\simeq}$.

\item
If ${\simeq} \subsetneq {\equiv}$, then $ {\simeq} \subseteq {\simeq^{\equiv}} \subsetneq {\equiv}$.
\end{enumerate}
\end{lemma}

\begin{proof}
\begin{enumerate}
\item
At first, we show that ${\simeq} \subseteq {\simeq^{\equiv}}$ for every ${\equiv} \supseteq {\simeq}$. As a special case, we have ${\simeq} \subset {\simeq^{\equiv}}$. Assume that $\alpha \simeq \beta$, then we can check that conditions in Definition~\ref{def:dec_bisimulation_expansion}, taking $\mathcal{R} = {\simeq}$. This is a routine work, by applying the Computation Lemma (Lemma~\ref{computation-lemma}).
To see  why ${\simeq^{\simeq}} \subseteq {\simeq}$, we notice ${\simeq^{\equiv}} \subseteq {\equiv}$ (Lemma~\ref{lem:decreasing_bisimulation_inequality}), and take ${\equiv}$ to be $\simeq$.

\item
By the proof of the first item and Lemma~\ref{lem:decreasing_bisimulation_inequality}, we already have $ {\simeq} \subseteq {\simeq^{\equiv}} \subseteq {\equiv}$ whenever ${\simeq} \subseteq {\equiv}$. Now we assume further that ${\simeq} \subseteq {\simeq^{\equiv}} = {\equiv}$, we will show that ${\simeq} = {\simeq^{\equiv}} = {\equiv}$. It suffices to show $ {\simeq^{\equiv}} = {\equiv}$ is a branching bisimulation (Definition~\ref{def:beq}).  Because ${\simeq^{\equiv}}$ is a decreasing bisimulation with expansion of $\equiv$, it satisfies the conditions in Definition~\ref{def:dec_bisimulation_expansion}. By taking $\mathcal{R}$ to be both ${\simeq^{\equiv}}$ and $\equiv$, we see that bisimulation property in Definition~\ref{def:beq} can be inferred.  \qed
\end{enumerate}
\end{proof}

The unique decomposition property of $\simeq^{\equiv}$ can be established in the same way as that of $\simeq$, but relies on the right cancellation property of ${\equiv}$.

\begin{theorem}[Unique Decomposition Property of $\simeq^{\equiv}$]\label{thm:unique-decomposition-relative}
Let ${\equiv}$ be a norm-preserving congruence which is right-cancellative. Then,
$\simeq^{\equiv}$ is decompositional.
\end{theorem}

\begin{proof}
It suffices to show that to show  that
$\{(\alpha, \beta): \alpha\gamma \simeq^{\equiv} \beta\gamma \mbox{ for some }\gamma\}$
is a decreasing branching bisimulation wrt.~$\equiv$.
In the proof the right cancellativity of ${\equiv}$ is used.  Then the proof goes in the same way as in Theorem~\ref{thm:unique-decomposition}.  \qed
\end{proof}

According to Theorem~\ref{thm:unique-decomposition-relative}, $\simeq^{\equiv}$ is decompositional whenever $\equiv$ is.  This is the key property to define refinement operation. Now, our refinement operation $\mathsf{Ref}(\equiv)$ can be defined as $\simeq^{\equiv}$.

\section{The Correctness of the Algorithm}\label{sec:correctness}

In this section we will show that the $\mathcal{B}'$ constructed from $\mathcal{B}$ during an iteration is exactly the decomposition base of $\mathsf{Ref}(\stackrel{\mathcal{B}}{\equiv}) = {\simeq^{\stackrel{\mathcal{B}}{\equiv}}}$ defined in Section~\ref{subsec:Expansion_in_general}.

\subsection{The Characterization of ${\simeq^{\equiv}}$}\label{subsec:Expansion_characterization}

Remember in Section~\ref{subsec:Expansion_Another_undestanding} we have remarked that the procedure in Fig.~\ref{Checking_EXP_REALTIME} is correct for realtime systems. At that time the proof is straightforward, because the procedure checks exactly the conditions in the characterization theorem (Theorem~\ref{lem:char_relative_bisimilarity_realtime}), which are exactly the conditions in Definition~\ref{def:dec_bisimulation_expansion_realtime}.  However, this is not the case now, and there are a number of subtleties.

In the following, we will develop some terminologies, which make us easier to formulate our results. First we need an adequate notion of `expansion' relation which is suitable for Definition~\ref{def:dec_bisimulation_expansion} and close to the testing procedure.  We call this notion {\em compound expansion}.

\begin{definition}\label{def:compound_expansion}
Let $\equiv$ be a norm-preserving congruence on processes, and
let ${\mathcal{R}} \subseteq {\equiv}$ be a relation on processes. The {\em compound expansion} wrt.~$\mathcal{R}$ and $\equiv$, denoted by $\mathsf{ComExp}_{\equiv}(\mathcal{R})$, contains all pairs
$(\alpha, \beta)$ which satisfy $\alpha \equiv \beta$ and  the following conditions:
\begin{enumerate}
\item
Whenever $\alpha \stackrel{\ell}{\longrightarrow} \alpha'$, then either
 \begin{enumerate}
 \item
  $\ell = \tau$ and $\alpha' \mathcal{R} \beta$; or

  \item
  $\beta \stackrel{\ell}{\longrightarrow}_{\mathrm{dec}} \beta'$ and $\alpha'\mathcal{R} \beta'$  for some $\beta'$; or

 \item
  $\beta \stackrel{\ell}{\longrightarrow}_{\mathrm{inc}} \beta'$ and $\alpha' \equiv \beta'$  for some $\beta'$; or

  \item
  $\beta \stackrel{\tau}{\longrightarrow}_{\mathrm{dec}} \beta''$ and $\alpha \mathcal{R} \beta''$  for some $\beta''$.
\end{enumerate}

\item
Whenever  $\beta \stackrel{\ell}{\longrightarrow} \beta'$, then either
 \begin{enumerate}
 \item
    $\ell=\tau$ and $\alpha \mathcal{R} \beta'$; or

 \item
  $\alpha \stackrel{\ell}{\longrightarrow}_{\mathrm{dec}} \alpha'$ and $\alpha'\mathcal{R} \beta'$  for some $\alpha'$; or

 \item
  $\alpha \stackrel{\ell}{\longrightarrow}_{\mathrm{inc}} \alpha'$ and $\alpha' \equiv \beta'$  for some $\alpha'$; or

  \item
  $\alpha \stackrel{\tau}{\longrightarrow}_{\mathrm{dec}} \alpha''$ and $\alpha'' \mathcal{R} \beta$  for some $\alpha''$.
   \end{enumerate}
\end{enumerate}
\end{definition}
The correctness of Definition~\ref{def:compound_expansion} is confirmed by the following lemmas.
\begin{lemma}\label{lem:decreasing_branching_bisimulation_oneside_contain}
If $\mathcal{R}$ is a decreasing bisimulation with expansion of $\equiv$ (see Definition~\ref{def:dec_bisimulation_expansion}),
then $\mathcal{R} \subseteq \mathsf{ComExp}_{\equiv}(\mathcal{R})$. In particular,
 ${\simeq^{\equiv}}  \subseteq \mathsf{ComExp}_{\equiv}(\simeq^{\equiv})$.
\end{lemma}

\begin{proof}
This fact is an inference of  Definition~\ref{def:compound_expansion} and Definition~\ref{def:dec_bisimulation_expansion}.  Compare the conditions in these two definitions. When $\mathcal{R}$ is a decreasing bisimulation with expansion of $\equiv$  and $\alpha \mathcal{R} \beta$, $(\alpha,\beta)$ satisfies the conditions in Definition~\ref{def:dec_bisimulation_expansion}. Then we can find that $(\alpha,\beta)$ also satisfies the conditions in Definition~\ref{def:compound_expansion}. \qed
\end{proof}

\begin{lemma}\label{lem:decreasing_branching_bisimulation_twoside_contain}
$\mathsf{ComExp}_{\equiv}({\simeq^{\equiv}})$ is a decreasing bisimulation with expansion of $\equiv$. In particular, $  \mathsf{ComExp}_{\equiv}(\simeq^{\equiv}) \subseteq {\simeq^{\equiv}}$.
\end{lemma}

\begin{proof}
At first, remember that  $\mathsf{ComExp}_{\equiv}(\mathcal{R}) \subseteq {\equiv}$ according to Definition~\ref{def:compound_expansion}. This is the prerequisite of $\mathsf{ComExp}_{\equiv}(\mathcal{R})$ being a decreasing bisimulation with expansion of $\equiv$.   This fact will be implicitly used in the remaining proof.

Let $(\alpha, \beta) \in \mathsf{ComExp}_{\equiv}({\simeq^{\equiv}})$ and $\alpha \stackrel{\ell}{\longrightarrow} \alpha'$.    According to the definition of $\mathsf{ComExp}_{\equiv}$ (Definition~\ref{def:compound_expansion}), there are four cases:
 \begin{enumerate}
 \item
  $\ell = \tau$ and $\alpha' \simeq^{\equiv} \beta$.  In this case,  we have $(\alpha', \beta) \in \mathsf{ComExp}_{\equiv}({\simeq^{\equiv}})$ according to Lemma~\ref{lem:decreasing_branching_bisimulation_oneside_contain}. Thus condition~1(a) of Definition~\ref{def:dec_bisimulation_expansion} holds (with $m=0$).

  \item
  $\beta \stackrel{\ell}{\longrightarrow}_{\mathrm{dec}} \beta'$ and $\alpha' \simeq^{\equiv} \beta'$  for some $\beta'$. In this case, we have $(\alpha', \beta') \in \mathsf{ComExp}_{\equiv}({\simeq^{\equiv}})$  according to  Lemma~\ref{lem:decreasing_branching_bisimulation_oneside_contain}. This is the special case of the condition~1(b) of Definition~\ref{def:dec_bisimulation_expansion} in which $m = 0$. Thus condition~1(b) of Definition~\ref{def:dec_bisimulation_expansion} holds.

  \item
  $\beta \stackrel{\ell}{\longrightarrow}_{\mathrm{inc}} \beta'$ and $\alpha' \equiv \beta'$  for some $\beta'$. This is the special case of the condition~1(c) of Definition~\ref{def:dec_bisimulation_expansion} in which $m = 0$. Thus condition~1(c) of Definition~\ref{def:dec_bisimulation_expansion} holds.

  \item
  $\beta \stackrel{\tau}{\longrightarrow}_{\mathrm{dec}} \beta''$ and $\alpha \simeq^{\equiv} \beta''$  for some $\beta''$. In this case, we have $(\alpha, \beta'') \in \mathsf{ComExp}_{\equiv}({\simeq^{\equiv}})$  according to  Lemma~\ref{lem:decreasing_branching_bisimulation_oneside_contain}.  We can now use induction hypothesis on the pair $(\alpha, \beta'')$.  Note that this case can not happen forever.  Finally, case~1 or case~2 or case~3 must happen.
   \begin{itemize}
    \item
   If case~1 happens finally, then we have $\beta  \stackrel{\tau}{\longrightarrow}_{\mathrm{dec}} \beta_1 \stackrel{\tau}{\longrightarrow}_{\mathrm{dec}} \ldots \stackrel{\tau}{\longrightarrow}_{\mathrm{dec}} \beta_i \stackrel{\tau}{\longrightarrow}_{\mathrm{dec}} \ldots \stackrel{\tau}{\longrightarrow}_{\mathrm{dec}} \beta_m$ for some $m > 0$ and $\beta_1,\ldots,\beta_m$  such that $\alpha' \simeq^{\equiv} \beta_m$ and $(\alpha, \beta_k)\in \mathsf{ComExp}_{\equiv}({\simeq^{\equiv}})$ for every $1 \leq i \leq m$.  Now we also have $(\alpha', \beta_m) \in \mathsf{ComExp}_{\equiv}({\simeq^{\equiv}})$ according to  Lemma~\ref{lem:decreasing_branching_bisimulation_oneside_contain}.
   Consequently condition~1(a) of Definition~\ref{def:dec_bisimulation_expansion} in which $m > 0$ holds.

  \item
   If case~2 happens finally, then we get $\beta  \stackrel{\tau}{\longrightarrow}_{\mathrm{dec}} \beta_1 \stackrel{\tau}{\longrightarrow}_{\mathrm{dec}} \ldots \stackrel{\tau}{\longrightarrow}_{\mathrm{dec}} \beta_i \stackrel{\tau}{\longrightarrow}_{\mathrm{dec}} \ldots \stackrel{\tau}{\longrightarrow}_{\mathrm{dec}} \beta_m \stackrel{\ell}{\longrightarrow}_{\mathrm{dec}} \beta'$ for some $m > 0$ and $\beta_1,\ldots,\beta_m$ and $\beta'$ such that $(\alpha' , \beta') \mathsf{ComExp}_{\equiv}({\simeq^{\equiv}})$ and $(\alpha, \beta_k) \in \mathsf{ComExp}_{\equiv}({\simeq^{\equiv}})$ for every $1 \leq i \leq m$, according to  Lemma~\ref{lem:decreasing_branching_bisimulation_oneside_contain}.
   Consequently condition~1(b) of Definition~\ref{def:dec_bisimulation_expansion} in which $m > 0$ holds.

 \item
    If case~3 happens finally, then we get $\beta  \stackrel{\tau}{\longrightarrow}_{\mathrm{dec}} \beta_1 \stackrel{\tau}{\longrightarrow}_{\mathrm{dec}} \ldots \stackrel{\tau}{\longrightarrow}_{\mathrm{dec}} \beta_i \stackrel{\tau}{\longrightarrow}_{\mathrm{dec}} \ldots \stackrel{\tau}{\longrightarrow}_{\mathrm{dec}} \beta_m \stackrel{\ell}{\longrightarrow}_{\mathrm{inc}} \beta'$ for some $m > 0$ and $\beta_1,\ldots,\beta_m$ and $\beta'$ such that $\alpha' \equiv \beta'$ and $(\alpha, \beta_k) \in \mathsf{ComExp}_{\equiv}({\simeq^{\equiv}})$ for every $1 \leq i \leq m$, according to  Lemma~\ref{lem:decreasing_branching_bisimulation_oneside_contain}.
   Consequently condition~1(c) of Definition~\ref{def:dec_bisimulation_expansion} in which $m > 0$ holds. \qed
  \end{itemize}
\end{enumerate}
\end{proof}
From Lemma~\ref{lem:decreasing_branching_bisimulation_oneside_contain} and  Lemma~\ref{lem:decreasing_branching_bisimulation_twoside_contain}, we conclude the following important characterization of $\simeq^{\equiv}$.

\begin{theorem}\label{thm:char_relative_bisimilarity}
$\alpha \simeq^{\equiv} \beta$ if and only if  $(\alpha, \beta) \in \mathsf{ComExp}_{\equiv}({\simeq^{\equiv}})$.
\end{theorem}

%def:dec_bisimulation_expansion
%def:compound_expansion

\begin{remark}
The inverse of Lemma~\ref{lem:decreasing_branching_bisimulation_oneside_contain} also holds. That is,
If a relation $\mathcal{R}$ satisfies $\mathcal{R} \subseteq \mathsf{ComExp}_{\equiv}(\mathcal{R})$, then $\mathcal{R}$ is a decreasing bisimulation with expansion of $\equiv$.    According to this fact and Theorem~\ref{thm:char_relative_bisimilarity}, the congruence ${\simeq^{\equiv}}$ is the greatest fixpoint of
$\mathsf{ComExp}_{\equiv}$. Thus the congruence ${\simeq^{\equiv}}$ can be completely characterized via the operation $\mathsf{ComExp}_{\equiv}$.

Readers may have noticed that the conditions in Definition~\ref{def:compound_expansion} are quite different from the conditions in Definition~\ref{def:dec_bisimulation_expansion}. In Definition~\ref{def:dec_bisimulation_expansion}
we lay stress on  getting a congruence relation from a congruence relation.    On the other hand,  in Definition~\ref{def:compound_expansion},  the purpose is to give a characterization which makes the conditions easy to check in the algorithm. We do not need $\mathsf{ComExp}_{\equiv}(\mathcal{R})$ to satisfy a lot of  favourite properties. The difference between these two definitions must be highlighted, because it does not happen in the case of realtime $\mathrm{nBPA}$, and the existence of silent actions do make things difficult. However, according to Theorem~\ref{thm:char_relative_bisimilarity}, $\mathsf{ComExp}_{\equiv}({\simeq^{\equiv}})$ is definitely a favourite congruence.
\end{remark}

\subsection{The Correctness of the Algorithm }\label{subsec:Expansion_correctness}

Theorem~\ref{thm:char_relative_bisimilarity} gives us a potential way to get an implementation of the refinement operation.   That is, it provides a potential way to implement $\mathtt{lpftest}_{(\mathbf{P}', \mathbf{E}')}(X_i, X_j)$ at line~10.

In the following discussion, for convenience we presuppose that $\stackrel{\mathcal{B}'}{\equiv}$ is equal to ${\simeq^{\stackrel{\mathcal{B}}{\equiv}}}$.  We will develop more properties of $\stackrel{\mathcal{B}'}{\equiv}$.

According to Theorem~\ref{thm:char_relative_bisimilarity}, checking  $X_i \stackrel{\mathcal{B}'}{\equiv} \delta$ is equivalent to checking $(X_i, \delta) \in \mathsf{ComExp}_{\stackrel{\mathcal{B}}{\equiv}}(\stackrel{\mathcal{B}'}{\equiv})$.
Note at first that $ \mathsf{ComExp}_{\stackrel{\mathcal{B}}{\equiv}}(\stackrel{\mathcal{B}'}{\equiv})$ concerns relation $\mathcal{B}'$ itself, and $\mathcal{B}'$ is not completely known at the moment. Fortunately, we have the following two critical observations.
\begin{description}
\item[Observation~1.]
At the moment of testing on the pair $(X_i,\delta)$,  we have already known the base $\mathcal{B}$ and a profile of  $\mathcal{B}'$ whose constances with indexes less than $i$. Thus we can suppose that $\mathtt{dcmp}_{\mathcal{B}} (X_i)$ is known for every $i$ such that $1 \leq i \leq n$, and $\mathtt{dcmp}_{\mathcal{B}'} (X_j)$ is known for every $j$ such that $1 \leq j < i$.  Therefore we are capable to answer  whether $\alpha \stackrel{\mathcal{B}}{\equiv} \beta$ for any $\alpha,\beta \in \mathbf{C}^{*}$, and whether $\alpha \stackrel{\mathcal{B}'}{\equiv} \beta$ for any $\alpha,\beta \in \mathbf{C}_{i-1}^{*}$.

\item[Observation~2.]
 Whenever decreasing transitions are concerned, say $X_i \stackrel{\ell}{\longrightarrow}_{\mathrm{dec}} \beta$,  according to Lemma~\ref{lem:decreasing_transition}, we have $\beta \in \mathbf{C}_{i-1}^{*}$.
\end{description}
With these two observations, we can develop the efficient procedure for checking $X_i \stackrel{\mathcal{B}'}{\equiv} \delta$ for realtime system (remember Theorem~\ref{lem:char_relative_bisimilarity_realtime}).

But at present, the situation is more complicated. In the presence of silent actions,
the above two observations cannot directly lead to the efficient procedure.   The consecutive silent actions do cause inconvenience. Investigate the following scenario. Assume we want to show $X_i \stackrel{\mathcal{B}'}{\equiv} \delta$ and let $X_i \stackrel{\ell}{\longrightarrow} \alpha$ be a transition which is required to be matched by $\delta \stackrel{\tau}{\longrightarrow}_{\mathrm{dec}} \beta'' \stackrel{\tau}{\longrightarrow}_{\mathrm{dec}} \ldots \stackrel{\ell}{\longrightarrow} \beta$ with $X_i \stackrel{\mathcal{B}'}{\equiv} \beta''$. In this situation we still do not know  whether $X_i \stackrel{\mathcal{B}'}{\equiv} \beta''$ because $\mathtt{dcmp}_{\mathcal{B}'} (X_i)$ still needs computing.

To handle this difficulty, we need some other techniques. Before doing this,  we notice the following critical observation:

\begin{description}
\item[Observation~3.]
According to the fact of $\delta \in \mathbf{P}'^{*}$ and Lemma~\ref{lem:tau_prime_string}, $\delta$ has no transition of the form $\delta \stackrel{\tau}{\longrightarrow} \beta'$ which satisfies $\delta \stackrel{\mathcal{B}'}{\equiv} \beta'$.
\end{description}

Whenever $X_i \stackrel{\mathcal{B}'}{\equiv} \delta$,
the above critical observation gives rise to the following lemma.
\begin{lemma}\label{lem:critical}
Assume ${\stackrel{\mathcal{B}'}{\equiv}} = {\simeq^{\stackrel{\mathcal{B}}{\equiv}}}$.  When $X_i \stackrel{\mathcal{B}'}{\equiv} \delta$ and  $\delta \stackrel{\tau}{\longrightarrow}_{\mathrm{def}} \beta$, then  we do \textbf{not} have $X_i \stackrel{\mathcal{B}'}{\equiv} \beta$.
\end{lemma}
According to Lemma~\ref{lem:critical}, we can draw the following two assertions.
First, when transition $\delta \stackrel{\tau}{\longrightarrow}_{\mathrm{def}} \beta$ is matched by $X_i$, the vacantly matching cannot happen.
Second, when transition $X_i \stackrel{\ell}{\longrightarrow} \alpha$ is matched by $\delta$, the `state-preserving' silent transitions cannot occurred.

Within these two assertions, Theorem~\ref{thm:char_relative_bisimilarity} can be written as follows.

\begin{theorem}\label{thm:exp_relative_membership_simplified}
Let $\mathcal{B} = (\mathbf{P}, \mathbf{E})$ and $\mathcal{B}' = (\mathbf{P}', \mathbf{E}')$ be two decomposition bases  which validate  ${\stackrel{\mathcal{B}'}{\equiv}} = {\simeq^{\stackrel{\mathcal{B}}{\equiv}}}$.    Assume
$\delta \in {\mathbf{P}'}_{i-1}^{*}$, then  $X_i \stackrel{\mathcal{B}'}{\equiv} \delta$ if and only if $X_i \stackrel{\mathcal{B}}{\equiv} \delta$ and the following conditions are satisfied:
\begin{enumerate}
\item
Whenever $X_i \stackrel{\ell}{\longrightarrow} \alpha$, then either
\begin{enumerate}
\item
$\ell = \tau$ and  $\alpha  \stackrel{\mathcal{B}'}{\equiv} \delta$; or

\item
$\delta \stackrel{\ell}{\longrightarrow}_{\mathrm{dec}}   \beta$ and $\alpha \stackrel{\mathcal{B}'}{\equiv} \beta$ for some $\beta$; or

\item
$\delta \stackrel{\ell}{\longrightarrow}_{\mathrm{inc}}   \beta$ and $\alpha \stackrel{\mathcal{B}}{\equiv} \beta$ for some $\beta$.
\end{enumerate}

\item
Either $X_i \stackrel{\tau}{\longrightarrow} \alpha$ and $\alpha \stackrel{\mathcal{B}'}{\equiv} \delta$ for some $\alpha$;

or, whenever $\delta \stackrel{\ell}{\longrightarrow} \beta$, either
\begin{enumerate}
\item
 $X_i \stackrel{\ell}{\longrightarrow}_{\mathrm{dec}}  \alpha$ and $\alpha \stackrel{\mathcal{B}'}{\equiv} \beta$ for some $\alpha$, or

\item
$X_i \stackrel{\ell}{\longrightarrow}_{\mathrm{inc}}  \alpha$ and $\alpha \stackrel{\mathcal{B}}{\equiv} \beta$ for some $\alpha$.
\end{enumerate}
\end{enumerate}
\end{theorem}

\begin{proof}
Remember that Theorem~\ref{thm:char_relative_bisimilarity} confirms that $X_i \stackrel{\mathcal{B}'}{\equiv} \delta$ if and only if  $(X_i, \delta) \in \mathsf{ComExp}_{\equiv}(\stackrel{\mathcal{B}'}{\equiv})$.  According to Definition~\ref{def:compound_expansion},
$X_i \stackrel{\mathcal{B}'}{\equiv} \delta$ if and only if $X_i \stackrel{\mathcal{B}}{\equiv} \delta$ and:
\begin{enumerate}
\item
Whenever $X_i \stackrel{\ell}{\longrightarrow} \alpha$, then either
 \begin{enumerate}
 \item
  $\ell = \tau$ and $\alpha \stackrel{\mathcal{B}'}{\equiv} \delta$; or

  \item
  $\delta \stackrel{\ell}{\longrightarrow}_{\mathrm{dec}} \beta$ and $\alpha \stackrel{\mathcal{B}'}{\equiv} \beta$  for some $\beta$; or

 \item
  $\delta \stackrel{\ell}{\longrightarrow}_{\mathrm{inc}} \beta$ and $\alpha' \stackrel{\mathcal{B}}{\equiv} \beta'$  for some $\beta$; or

  \item
  $\delta \stackrel{\tau}{\longrightarrow}_{\mathrm{dec}} \beta'$ and $X_i \stackrel{\mathcal{B}'}{\equiv} \beta'$  for some $\beta'$.
\end{enumerate}

\item
Whenever  $\delta \stackrel{\ell}{\longrightarrow} \beta$, then either
 \begin{enumerate}
 \item
    $\ell=\tau$ and $X_i \stackrel{\mathcal{B}'}{\equiv} \beta$; or

 \item
  $X_i \stackrel{\ell}{\longrightarrow}_{\mathrm{dec}} \alpha$ and $\alpha \stackrel{\mathcal{B}'}{\equiv} \beta$  for some $\alpha$; or

 \item
  $X_i \stackrel{\ell}{\longrightarrow}_{\mathrm{inc}} \alpha$ and $\alpha \stackrel{\mathcal{B}}{\equiv} \beta$  for some $\alpha$; or

  \item
  $X_i \stackrel{\tau}{\longrightarrow}_{\mathrm{dec}} \alpha'$ and $\alpha' \stackrel{\mathcal{B}'}{\equiv} \beta$  for some $\alpha'$.
   \end{enumerate}
\end{enumerate}
Now making use of Lemma~\ref{lem:critical}, we can draw the conclusion that the case~1(d) and case~2(a) cannot happen! Now the conditions above become the conditions in Theorem~\ref{thm:exp_relative_membership_simplified}. \qed
\end{proof}

Comparing with Theorem~\ref{thm:char_relative_bisimilarity}, Theorem~\ref{thm:exp_relative_membership_simplified} has a great advantage.  When we need to determine whether $X_i  \stackrel{\mathcal{B}'}{\equiv} \delta$ or not, according to Theorem~\ref{thm:exp_relative_membership_simplified}, we only require to checking several conditions which depends only on $\mathcal{B}$ and the profile of $\mathcal{B}'$ in which only constants with index less than $i$ are involved. Thus we can use this fact to construct $\mathcal{B}'$ in the `bottom-up' way, which is exactly the procedure described in Fig.~\ref{Checking_EXP}. The proof of correctness of the algorithm is now finished.

\section{Remark}\label{sec:remark}

\subsection{Other Bisimilarities On Totally Normed $\mathrm{BPA}$}
Comparing with branching bisimilarity, other bisimilarities tend to be more flexible so that they are currently known to be NP-hard on $\mathrm{tnBPA}$.  On the occasion of weak bisimilarity, there are two different problems deserving to consideration. First, it is no longer decompositional, as is shown in Example~\ref{example:weak_bisimilarity_decom}.  Second, it is capable to encode NP-complete problem due to its more flexible matching style.

There is a variant of weak bisimilarity called delay bisimilarity, which is still decompositional on $\mathrm{tnBPA}$. Using unique decomposition property, we can confirm that delay bisimilarity is in PSPACE. The way is barely to guess a decomposition base $\mathcal{B} = (\mathbf{P}, \mathbf{E})$ and check that $\stackrel{\mathcal{B}}{\equiv}$ a delay bisimulation.  Still, the bisimulation property needs to carefully defined.  Anyway, it is technically much easier than checking branching bisimilarity.

Finally we conjecture that deciding bisimilarities other than branching bisimilarity on tnBPA is PSPACE complete.

\subsection{On Branching Bisimilarity Checking}

In the situation that silent transitions are treated unobservable, branching bisimilarity arouses interest of researchers. In most of the cases, previous decidability and complexity results for weak bisimilarity still hold for branching bisimilarity. There are two remarkable exceptions. The decidability of branching bisimilarity is established by Czerwi\'{n}ski, Hofman and Lasota~\cite{DBLP:conf/concur/CzerwinskiHL11} on normed BPP, and by Fu~\cite{DBLP:conf/icalp/Fu13} on  normed BPA.  In these two cases, decidability of weak bisimilarity is unknown.  Recently, we have proven that branching (and weak) bisimilarity is undecidable on every model above BPA and BPP in the PRS hierarchy even in the normed case~\cite{DBLP:conf/icalp/YinFHHT14}. It is believed that branching bisimilarity is easier to decide than weak bisimilarity. Currently, there is no real instance to support this belief. This paper provides an interesting instance. We expect that more instances will be discovered in the future.

%________________________________________________________
\vspace*{4.5mm}\noindent {\bf Acknowledgement}.  The author would like to thank S{\l}awomir Lasota for letting me know the work of Czerwi\'{n}ski~\cite{CzerwinskiPhD}, the current fastest algorithm for checking strong bisimilarity on normed $\mathrm{BPA}$;
%to thank Bas Luttik for the simplest counterexample;
to thank the members of BASICS for their helpful discussions on related topics.
%The support from NSFC (61033002, ANR 61261130589) is gratefully acknowledged.

%________________________________________________________

\bibliographystyle{plain}
\bibliography{arxiv}

%________________________________________________________

\newpage

\end{document}